\documentclass[journal]{IEEEtran}

\usepackage{cite}
\ifCLASSINFOpdf
   \usepackage[pdftex]{graphicx}
\else
  \usepackage[dvips]{graphicx}
\fi

\usepackage[cmex10]{amsmath}
\usepackage{amssymb}
\usepackage{amsfonts}
\usepackage{amsthm}
\usepackage{array}
\usepackage{dsfont}
\usepackage{epstopdf}
\usepackage{balance}

\newtheorem{lemma}{Lemma}

\newtheorem{theorem}{Theorem}

\newtheorem{proposition}{Proposition}
\newtheorem{definition}{Definition}
\long\def\symbolfootnote[#1]#2{\begingroup%
\def\thefootnote{\fnsymbol{footnote}}\footnote[#1]{#2}\endgroup}

\begin{document}


\title{Capacity of a Class of State-Dependent Orthogonal Relay Channels}


\author{
  \IEEEauthorblockN{I\~naki Estella Aguerri and Deniz G\"{u}nd\"{u}z
  \\}
\IEEEauthorblockA{
\thanks{ 
This paper was presented in part at the IEEE Information Theory Workshop (ITW), Laussane, Switzerland, Sept. 2012  \cite{Estella2012CapacityClassState}.
I. Estella Aguerri is with the Mathematical and Algorithmic Sciences
Laboratory, France Research Center, Huawei Technologies Company, Ltd.,
Boulogne-Billancourt 92100, France E-mail: {\tt inaki.estella@huawei.com\tt}. 
Deniz G\"{u}nd\"{u}z is with the Department of Electrical and Electronic
Engineer, Imperial College London, London SW7 2BT, U.K. E-mail: {\tt  d.gunduz}{\tt @imperial.ac.uk}.
The research presented here was done during the Ph.D. studies of I\~naki Estella Aguerri at Imperial College London. }
}
}

\maketitle

\IEEEpeerreviewmaketitle

\begin{abstract}
The class of orthogonal relay channels in which the orthogonal channels connecting the source terminal to the relay and the destination, and the relay to the destination, depend on a state sequence, is considered. It is assumed that the state sequence is fully known at the destination while it is not known at the source or the relay. The capacity of this class of relay channels is characterized, and  shown to be achieved by the partial decode-compress-and-forward (pDCF) scheme. Then the capacity of certain binary and Gaussian state-dependent orthogonal relay channels are studied in detail, and it is shown that the compress-and-forward (CF) and partial-decode-and-forward (pDF) schemes are suboptimal in general. To the best
of our knowledge, this is the first single relay channel model for
which the capacity is achieved by pDCF, while
pDF and CF schemes are both suboptimal. Furthermore, it is shown that the capacity of the considered  class of state-dependent orthogonal relay channels is in general below the cut-set bound. The conditions under which pDF or CF suffices to meet the cut-set bound, and hence, achieve the capacity, are also derived. 
\end{abstract}

\begin{IEEEkeywords} Capacity, channels with state, relay channel, decode-and-forward, compress-and-forward, partial decode-compress-and forward.
\end{IEEEkeywords}

\section{Introduction}

We consider a state-dependent orthogonal relay channel, in which the channels connecting the source to the relay, and the source and the relay to the destination are orthogonal, and are governed by a state sequence, which is assumed to be known only at the destination. We call this model the \emph{state-dependent orthogonal relay channel with state information available at the destination}, and refer to it as the ORC-D model. See Figure \ref{fig:RelayChannelR0R1} for an illustration of the ORC-D channel model.

Many practical communication scenarios can be modelled by the ORC-D model. For example, consider a cognitive network with a relay, in which the transmit signal of the secondary user interferes simultaneously with the received primary user signals at both the relay and the destination. After decoding the secondary user's message, the destination obtains information about the interference affecting the source-relay channel, which can be exploited to decode the primary transmitter's message. Note that the relay may be oblivious to the presence of the secondary user, and hence, may not have access to the side information. Similarly, consider a mobile network with a relay (e.g., a femtostation), in which the base station (BS) operates in the full-duplex mode, and transmits on the downlink channel to a user, in parallel to the uplink transmission of a femtocell user, causing interference for the uplink transmission at the femtostation. While the femtostation, i.e., the relay, has no prior information about this interfering signal, the BS knows it perfectly and can exploit this knowledge to decode the uplink user's message forwarded by the femtostation.

The best known transmission strategies for the three terminal relay channel are the decode-and-forward (DF), compress-and-forward (CF) and partial decode-compress-and-forward (pDCF) schemes, which were all introduced by Cover and El Gamal in \cite{Cover:1979}. In DF, the relay decodes the source message and forwards it to the destination together with the source terminal. DF is generalized by the partial decode-and-forward (pDF) scheme in which the relay decodes and forwards only a part of the message. In the ORC-D model, pDF would be optimal when the channel state information is not available at the destination
\cite{Gamal2005OrthRelay}; however, when the state information is known at the destination, fully decoding and re-encoding the message transmitted on the source-relay link renders the channel state information at the destination useless. Hence, we expect that pDF is suboptimal for ORC-D in general.

In CF, the relay does not decode any part of the message, and
simply compresses the received signal and forwards the compressed bits to the destination using Wyner-Ziv coding followed by separate channel coding. Using CF in the ORC-D model allows the destination to exploit its knowledge of the state sequence; and hence, it can decode messages that may not be decodable by the relay. However, CF also forwards some noise to the destination, and therefore, may be suboptimal in certain scenarios. For example, as the dependence of the source-relay channel on the state sequence weakens, i.e., when the state information becomes less informative, CF performance is expected to degrade.

pDCF combines both schemes: part of the source message is
decoded by the relay, and forwarded, while the remaining signal is compressed and forwarded to the destination. Hence, pDCF can optimally adapt its transmission to the dependence of the orthogonal channels on the state sequence. Indeed, we show that  pDCF achieves the capacity in the ORC-D channel model, while pure DF and CF are in general suboptimal. The main results of the paper are summarized as follows:

\begin{itemize}
\item We derive an upper bound on the capacity of the ORC-D model, and show that it is achievable by the pDCF scheme. This characterizes the capacity of this class of relay channels.
\item  Focusing on the multi-hop binary and Gaussian models, we show that applying either only the CF or only the DF scheme is in general suboptimal.
\item We show that the capacity of the ORC-D model is in general below the cut-set bound. We identify the conditions under which pure DF or pure CF meet the cut-set bound. Under these conditions the cut-set bounds is tight, and either DF or CF scheme is sufficient to achieve the capacity.
\end{itemize}

While the capacity of the general relay channel is still an open
problem, there have been significant achievements within the last decade in understanding the capabilities of various transmission schemes, and the capacity of some classes of relay channels has been characterized. For example, DF is shown to be optimal for physically
degraded relay channels and inversely degraded relay channels in
\cite{Cover:1979}. In \cite{Gamal2005OrthRelay}, the capacity of the
orthogonal relay channel is characterized, and shown to be achieved by the
pDF scheme. It is shown in \cite{ElGamal1982SemideterministiRelay}
that pDF achieves the capacity of semi-deterministic relay channels
as well. CF is shown to achieve the capacity in deterministic
primitive relay channels in \cite{Kim2007CapacityClassRelay}.  While
all of these capacity results are obtained by using the cut-set
bound for the converse proof \cite{Cover:book}, the capacity of a
class of modulo-sum relay channels is characterized in
\cite{Aleksic2009CapRelayModSum}, and it is shown that the capacity,
achievable by the CF scheme, can be below the cut-set bound.
The pDCF scheme is shown to achieve the capacity of a class of diamond relay channels in \cite{Kang2011DiamChan}.

The state-dependent relay channel has also attracted considerable attention in the literature. Key to the investigation of the state-dependent relay channel model is whether the state sequence controlling the channel is known at the nodes of the network, the source, relay or the destination in a causal or non-causal manner. The relay channel in which the state information is  non-causally available only at the source is considered in \cite{
Khormuji2009CoopDownlinFreqReuse, Zaidi2010:RelayNoncausalStateSource}, and both causally and non-causally available state information is considered in \cite{Akhbari2010AchRateRegionsRelayPrivateMEssState}. The model in which the state is non-causally known only at the relay  is studied in \cite{Zaidi2010:RelayNoncausalStateSource_Journ} while causal and non-causal knowledge is considered in  \cite{Akhbari2010CFstrategyCausalState}. Similarly, the relay channel with state causally known at the source and relay is considered in  \cite{Mirmohseni2009CFstrategyforRelaywithState}, and  state non-causally known at the source, relay and destination in \cite{Khormuji2013StateDependentRelay}.  Recently a generalization of pDF, called the cooperative-bin-forward scheme, has been shown to achieve the capacity of state-dependent semi-deterministic relay channels with causal state information at the source and destination \cite{Kolte205arxiv}. The compound relay channel with informed relay and destination  are discussed in \cite{Simeone2009:CompoundInformedRelay} and \cite{Behboodi2009:RelInfRec}. The state-dependent relay channel with structured state has been considered in \cite{Bakanoglu2012StructuredInterference} and \cite{Bakanoglu2011:HalfDupRelayInterf}. 
 To the best of our knowledge, this is the first work that focuses on the state-dependent relay channel in which the state information is available only at the destination.

The rest of the paper is organized as follows. In Section II we
provide the system model and our main result. Section III is devoted to the proofs of the achievability and converse for the main result. In Section IV, we provide two examples demonstrating the suboptimality of pDF and CF schemes on their own, and in Section V
we show that the capacity is in general below the cut-set bound, and we provide conditions under which pure DF and CF schemes
meet the cut-set bound. Finally, Section VII concludes the paper.

We use the following notation in the rest of the paper:
$X_i^j\triangleq(X_i, X_{i+1}, ..., X_j)$ for $i < j$,
$X^n\triangleq(X_1,  ..., X_n)$ for the complete sequence,
$X_{n+1}^n\triangleq \emptyset$, and $Z^{n\setminus
i}\triangleq(Z_1,...,Z_{i-1},Z_{i+1},...,Z_n)$.

\section{System Model and Main Result}\label{sec:Model}
We consider the class of orthogonal relay channels depicted in Figure \ref{fig:RelayChannelR0R1}. The
source and the relay are connected through a memoryless channel characterized by
$p(y_R|x_1,z)$, the source and the destination are connected through an orthogonal memoryless channel characterized by  $p(y_2|x_2,z)$, while the relay and the destination are
connected by a memoryless channel $p(y_1|x_R,z)$. The three memoryless channels depend on an independent and identically distributed (i.i.d.) state sequence $\{Z\}_{i=1}^n$, which is available at the destination. The input and output alphabets are denoted by $\mathcal{X}_1$,
$\mathcal{X}_2$, $\mathcal{X}_R$, $\mathcal{Y}_1$, $\mathcal{Y}_2$
and $\mathcal{Y}_R$, and the state alphabet is denoted by $\mathcal{Z}$.

\begin{figure}
\centering
\includegraphics[width=0.48\textwidth]{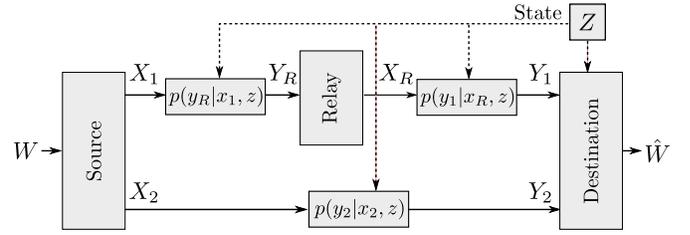}
\vspace{-3 mm} \caption{Orthogonal state-dependent relay channel with channel state
information available at the destination, called the ORC-D model.}
\label{fig:RelayChannelR0R1}
\end{figure}

Let $W$ be the message to be transmitted to the destination with the
assistance of the relay. The message $W$ is assumed to be uniformly
distributed over the set $\mathcal{W}=\{1,...,M\}$. An $(M,n,\nu_n)$
code for this channel consists of an encoding function at the
source:
\begin{IEEEeqnarray}{rCl}
f:\{1,...,M\}\rightarrow\mathcal{X}_1^n\times \mathcal{X}_2^n,
\end{IEEEeqnarray}
a set of encoding functions $\{f_{r,i}\}_{i=1}^n$ at the relay, whose output at time $i$ depends on the symbols it has received up to time $i-1$:
\begin{IEEEeqnarray}{rCl}
X_{Ri}=f_{r,i}(Y_{R1},...,Y_{R(i-1)}),\quad i=1,...,n,
\end{IEEEeqnarray}
and a decoding function at the destination
\begin{IEEEeqnarray}{rCl}
g:\mathcal{Y}_1^n\times\mathcal{Y}_2^n\times
\mathcal{Z}^n\rightarrow\{1,...,M\}.
\end{IEEEeqnarray}
The probability of error, $\nu_n$, is defined as
\begin{IEEEeqnarray}{rCl}
\nu_n\triangleq\frac{1}{M}\sum_{w=1}^M\text{Pr}\{g(Y_1^n,Y_2^n,Z^n)\neq
w|W=w\}.
\end{IEEEeqnarray}

The joint probability mass function (pmf) of the involved random variables
over the set
$\mathcal{W}\times\mathcal{Z}^n\times\mathcal{X}_1^n\times\mathcal{X}_2^n\times\mathcal{X}_R^n\times\mathcal{Y}_R^n\times\mathcal{Y}_1^n\times\mathcal{Y}_2^n$
is given by
\begin{IEEEeqnarray}{rCl}
p(w,z^n,x_1^n,x_2^n,x_R^n,y_R^n,y_1^n,y_2^n)\!=\!p(w)\prod_{i=1}^np(z_i)p(x_{1i},x_{2i}|w)\cdot\nonumber\\
 p(y_{Ri}|z_i,x_{1i})
 p(x_{Ri}|y_R^{i-1})p(y_{1i}|x_{Ri},z_i)p(y_{2i}|x_{2i},z_i)\nonumber.
\end{IEEEeqnarray}

A rate $R$ is said to be \textit{achievable} if there exists a sequence of
$(2^{nR},n,\nu_n)$ codes such that
$\lim_{n\rightarrow\infty}\nu_n=0$. The \textit{capacity}, $\mathcal{C}$, of
this class of state-dependent orthogonal relay channels, denoted as ORC-D, is
defined as the supremum of the set of all achievable rates.

We define $R_1$ and $R_1$ as follows, which can be thought as the capacities of the individual links from the relay to
the destination, and from the source to the destination, respectively, when the channel state
sequence is available at the destination:
\begin{IEEEeqnarray}{rCl}\label{eq:R0R1rates}
R_1\triangleq\max_{p(x_R)}I(X_R;Y_1|Z),\quad
R_2\triangleq\max_{p(x_2)}I(X_2;Y_2|Z).
\end{IEEEeqnarray}
Let $p^*(x_R)$ and $p^*(x_2)$ be the channel input distributions
achieving $R_1$ and $R_2$, respectively.

Let us define $\mathcal{P}$
as the set of all joint pmf's given by
\begin{IEEEeqnarray}{lll}\label{eq:pmf}
\mathcal{P}\triangleq\{p(u,x_1,z,y_R,\hat{y}_R):\\p(u,x_1,z,y_R,\hat{y}_R)= p(u,x_1)p(z)p(y_R|x_1,z)p(\hat{y}_R|y_R,u)\},\nonumber
\end{IEEEeqnarray}
where $U$ and $\hat{Y}_R$ are auxiliary random variables defined over the alphabets $\mathcal{U}$ and $\mathcal{\hat{Y}}_R$, respectively.

 The main result of this work, provided in the next
theorem, is the capacity of the class of  relay channels described
above.

\begin{theorem}\label{th:capreg}
The capacity of the ORC-D relay channel  is given by
\begin{IEEEeqnarray}{rCl}\label{eq:capreg}
\mathcal{C}=&\sup_{\mathcal{P}}\,&R_2+I(U;Y_R)+I(X_1;\hat{Y}_R|UZ),\nonumber\\
&\text{s.t. }&R_1\geq I(U;Y_R)+I(Y_R;\hat{Y}_R|UZ),
\end{IEEEeqnarray}
where $|\mathcal{U}|\leq |\mathcal{X}_1|+3$ and
$|\mathcal{\hat{Y}}_R|\leq |\mathcal{U}||\mathcal{Y}_R|+1$.
\end{theorem}
\begin{proof}
The achievability part of the theorem is proven in Section \ref{sec:Achiev}, while the converse proof can be found in Section \ref{sec:Converse}.
\end{proof}

In the next section, we show that the capacity of this class of state-dependent relay channels is achieved by the pDCF scheme. To the best
of our knowledge, this is the first single-relay channel model for
which the capacity is achieved by pDCF, while the pDF and CF schemes are
both suboptimal in general. In addition, the capacity of this relay channel is in general below the cut-set bound \cite{Cover:book}. These issues are discussed in more detail in Sections \ref{sec:Examples} and \ref{seq:cutset}.

It follows from Theorem \ref{th:capreg} that the transmission over the relay-destination and  source-destination links can be independently optimized to operate at the corresponding capacities, and these links in principle act as error-free channels of capacity $R_1$ and $R_2$, respectively. 
We also note that the relay can acquire some knowledge about the channel state sequence $Z^n$  from its channel output $Y^n_R$, and could use it in the transmission over the relay-destination link, which depends on the same state information sequence. In general, non-causal state information available at the relay can be exploited to increase the achievable throughput in multi-user setups \cite{MinSimeone2011StateRelay,Zaidi2010Cooperative:IT}. However, it follows from Theorem \ref{th:capreg} that this knowledge is useless. This is because the channel state information acquired from $Y^n_R$ can be seen as delayed feedback to the relay, which does not increase the capacity in point-to-point channels.

\subsection{Comparison with previous relay channel models}\label{ssec:ComparisonOrthogonal}
Here, we compare ORC-D with other relay channel models in the literature, and discuss the differences and similarities. The discrete memoryless relay channel consists of four finite sets $\mathcal{X}$, $\mathcal{X}_R$, $\mathcal{Y}$ and $\mathcal{Y}_R$, and a probability distribution $p(y,y_R|x,x_R)$. In this setup, $X$ corresponds to the source input to the channel, $Y$ to the channel output available at the destination, while $Y_R$ is the channel output available at the relay, and $X_R$ is the channel input symbol chosen by the relay. We note that the three-terminal relay channel model in \cite{Cover:1979} reduces to ORC-D by setting $X^n=(X_1^n,X_2^n)$, $Y^n=(Y_{1}^n,Y_{2}^n,Z^n)$, and $p(y,y_R|x_1x_R)=p(y_1,y_2,y_R,z|x_1,x_2,x_R)=p(z)p(y_R|x_1,z)p(y_1|x_R,z)p(y_2|x_2,z)$.

\begin{figure}
\centering
\includegraphics[width=0.5\textwidth]{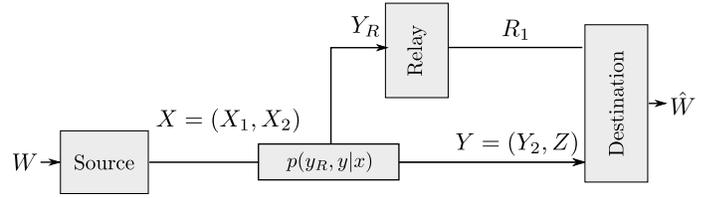}
\vspace{-3 mm} \caption{The ORC-D as a class of primitive relay channel.}
\label{fig:PrimitiveRelay}
\end{figure}

\begin{figure}
\centering
\includegraphics[width=0.5\textwidth]{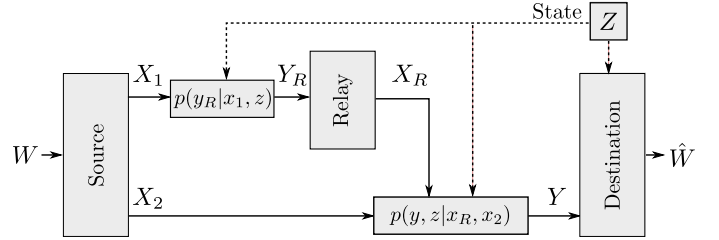}
\vspace{-3 mm} \caption{The ORC-D is a particular case of the state dependent orthogonal relay channel with orthogonal components.}
\label{fig:Orthogonal_Relay}
\end{figure}

By considering the channel from the relay to the destination as an  error-free link with finite capacity, the ORC-D is included in the class of primitive relay channels proposed in  \cite{Kim2007CapacityClassRelay} and \cite{Kim2007CodingPrimitiveRelay} as seen in Figure \ref{fig:PrimitiveRelay}, for which the channel distribution satisfies $p(y,y_R|x,x_R)=p(y,y_R|x)$. Although the capacity of this channel remains unknown in general, it has been characterized for certain special cases. CF has been shown to achieve the cut-set bound, i.e., to be optimal, in \cite{Kim2007CapacityClassRelay}, if the relay output, $Y_R$, is a deterministic function of the source input and output at the destination, i.e., $Y_R=f(X,Y)$. The capacity of a class of primitive relay channels under a particular modulo sum structure is shown to be achievable by CF in  \cite{Aleksic2009CapRelayModSum}, and to be below the cut-set bound. Theorem \ref{th:capreg} provides the optimality of pDCF for a class of primitive relay channels, not included in any of the previous relay models for which the capacity is known. 
It is discussed in \cite{Kim2007CodingPrimitiveRelay} that for the primitive relay channel, CF and DF do not outperform one another in general. It is also noted that their combination in the form of pDCF might not be sufficient to achieve the capacity  in general. We will see in Section \ref{sec:Examples} that both DF and CF are in general suboptimal, and that pDCF is necessary and sufficient to achieve the capacity for the class of primitive relay channels considered in this paper.

It is also interesting to compare the ORC-D model with the orthogonal relay channel proposed in \cite{Gamal2005OrthRelay}, in which the source-relay link is orthogonal to the multiple-access channel from the source and relay to the destination, i.e., $p(y,y_R|x,x_R)=p(y|x_2,x_R)p(y_R|x_1,x_R)$. The capacity for this model is shown to be achievable by pDF, and coincides with the cut-set bound. For the ORC-D, we have $p(y,y_R|x,x_R)=p(z)p(y_2|x_2,z)p(y_1|x_R,z)p(y_R|x_1,x_R,z)$, i.e., given the channel inputs, the orthogonal channel outputs at the relay and the destination are still dependent due to $Z$. Therefore, the ORC-D does not fall within the class of orthogonal channels considered in \cite{Gamal2005OrthRelay}. We can consider the class of state dependent relay channel with orthogonal components satisfying $p(y,z,y_R|x,x_R)=p(z)p(y|x_2,x_R,z)p(y_R|x_1,x_R,z)$  as shown in Figure \ref{fig:Orthogonal_Relay}. This class includes the orthogonal relay channel in \cite{Gamal2005OrthRelay} and the ORC-D as a particular cases. However, the capacity for this class of state dependent relay channel remains open in general.

\section{Proof of Theorem \ref{th:capreg}}\label{sec:capreg}
We first show in Section \ref{sec:Achiev} that
the capacity claimed in Theorem \ref{th:capreg} is achievable by pDCF. Then, we derive the converse result for Theorem \ref{th:capreg} in Section \ref{sec:Converse}.

\subsection{Achievability }\label{sec:Achiev}
We derive the rate achievable by the pDCF scheme for ORC-D using the achievable rate expression for the pDCF scheme proposed in \cite{Cover:1979} for the general relay channel. Note that the three-terminal relay channel in \cite{Cover:1979} reduces to  ORC-D by setting $X^n=(X_1^n,X_2^n)$ and $Y^n=(Y_{1}^n,Y_{2}^n,Z^n)$, as discussed in Section \ref{ssec:ComparisonOrthogonal}.

 In pDCF for the general relay channel, the source applies message splitting, and the
relay decodes only a part of the message. The part to be
decoded by the relay is transmitted through the auxiliary random
variable $U^n$, while the rest of the message is superposed onto
this through channel input $X^n$. Block Markov encoding is used for
transmission. The relay receives $Y_R^n$ and decodes only the part
of the message that is conveyed by $U^n$. The remaining signal
$Y_R^n$ is compressed into $\hat{Y}_R^n$. The decoded message is
forwarded through $V^n$, which is correlated with $U^n$, and the
compressed signal is superposed onto $V^n$ through the relay channel
input $X_R^n$. At the destination the received signal $Y^n$ is used
to recover the message. See \cite{Cover:1979} for details. The
achievable rate of the pDCF scheme is given below.

\begin{theorem}\label{th:pDCF}\textnormal{(Theorem 7,\cite{Cover:1979})}
The capacity of a relay channel $p(y,y_R|x,x_R)$ is lower bounded by
the following rate:
\begin{IEEEeqnarray}{rCl}\label{eq:GenPartialDFCF}
R_{\mathrm{pDCF}}=\sup&\min&\{I(X;Y,\hat{Y}_R|X_R,U)+I(U;Y_R|X_R,V),\nonumber\\
&&\;\;I(X,X_R;Y)-I(\hat{Y}_R;Y_R|X, X_R, U, Y)\},\nonumber\\
&\text{s.t. } & I(\hat{Y}_R;Y_R|Y, X_R, U)\leq I(X_R;Y|V),
\end{IEEEeqnarray}
where the supremum is taken over all joint pmf's of the form
\begin{IEEEeqnarray}{rCl}
p(v)p(u|v)p(x|u)p(x_R|v)p(y,y_R|x,x_R)p(\hat{y}_R|x_R,y_R,u).\nonumber
\end{IEEEeqnarray}
\end{theorem}

Since ORC-D is a special case of  the general relay channel model, the rate $R_{\mathrm{pDCF}}$ is achievable in an ORC-D as well.  The capacity achieving pDCF scheme for ORC-D  is obtained from (\ref{eq:GenPartialDFCF}) by setting $V=\emptyset$, and generating $X_R^n$ and $X_1^n$ independent of the rest of the variables with distribution $p^*(x_R)$ and $p^*(x_1)$,
respectively, as given in the next lemma.

\begin{lemma}\label{lem:Achi}
For the class of relay channels characterized by the ORC-D model, the capacity expression $\mathcal{C}$ defined in
(\ref{eq:capreg}) is achievable by the pDCF scheme.
\end{lemma}
\begin{proof}
See Appendix \ref{app:ProofLemmaAch}.
\end{proof}
The optimal pDCF scheme for ORC-D applies independent coding over the source-destination and the source-relay-destination branches. The source applies message splitting. Part of the message is transmitted over the source-destination branch and decoded at the destination using $Y_2^n$ and $Z^n$. In the relay branch, the part of the message to be decoded at the relay is transmitted through $U^n$, while the rest of the message is superposed onto this through the channel input $X_1^n$. At the relay the part conveyed by $U^n$ is decoded from $Y_R^n$, and the remaining signal $Y_R^n$  is compressed into $\hat{Y}_R^n$ using binning and assuming that $Z^n$ is available at the decoder. Both $U^n$ and the bin index corresponding to $\hat{Y}_R^n$ are transmitted over the relay-destination channel using $X_R^n$. At the destination, $X_R^n$ is decoded from $Y_1^n$, and $U^n$ and the bin index are recovered. Then, the decoder looks for the part of message transmitted over the relay branch jointly typical with $\hat{Y}_R^n$ within the corresponding bin and $Z^n$.

\subsection{Converse}\label{sec:Converse}

 The proof of the converse consists of two
parts. First we derive a single-letter upper bound on
the capacity, and then, we provide an alternative expression for this bound, which coincides with the rate achievable by pDCF.

\begin{lemma}\label{lem:Conv1}
The capacity of the class of relay channels characterized by the ORC-D model is upper bounded by
\begin{IEEEeqnarray}{rCl}\label{eq:Conv1}
R_{up}=\sup_{\mathcal{P}}\min\{&&R_2+I(U;Y_{R})+I(X_{1};\hat{Y}_{R}|UZ),\\ 
&& R_1+R_2-I(\hat{Y}_{R};Y_{R}|X_{1}UZ)\}.
\end{IEEEeqnarray}
\end{lemma}

\begin{proof}
See Appendix \ref{app:Converse1}.
\end{proof}

As stated in the next lemma, the upper bound $R_{up}$, given in Lemma \ref{lem:Conv1}, is equivalent to the capacity expression $\mathcal{C}$ given in Theorem \ref{th:capreg}. Since the achievable rate meets the upper bound, this concludes the proof of Theorem \ref{th:capreg}.

\begin{lemma}\label{lem:Conv2}
The upper bound on the achievable rate $R_{up}$ given in Lemma
\ref{lem:Conv1} is equivalent to the capacity expression $\mathcal{C}$ in Theorem \ref{th:capreg}.
\end{lemma}

\begin{proof}
See Appendix \ref{app:Converse2}.
\end{proof}

\section{The Multihop Relay Channel with State: Suboptimality of Pure pDF and CF schemes}\label{sec:Examples}

We have seen in Section \ref{sec:capreg} that the pDCF
scheme is capacity-achieving for the class of relay channels characterized by the ORC-D model. In order to prove the suboptimality of the pure DF and CF schemes for this class of relay channels, we consider a simplified system model, called the \emph{multihop relay channel with state information available at the destination} (MRC-D), which is obtained by simply removing the direct channel from the source to the destination,
i.e., $R_2=0$.

The capacity of this multihop relay channel model and the optimality of pDCF follows directly from Theorem \ref{th:capreg}. However, the  single-letter capacity expression depends on the joint pmf of $X_1$, $Y_R$, $X_R$ and $Y_1$ together with the auxiliary random variables $U$ and $\hat{Y}_R$. Unfortunately, the numerical characterization of the optimal joint pmf of these random variables is very complicated for most channels. A simple and computable upper bound on the capacity can be obtained from the
cut-set bound \cite{elGamal:book}. For MRC-D, the cut-set bound is given by
\begin{IEEEeqnarray}{rCl}\label{eq:cutset}
R_{\mathrm{CS}}=\min\{R_1,\max_{p(x_1)}I(X_1;Y_R|Z)\}.
\end{IEEEeqnarray}

Next, we characterize the rates achievable by the DF and CF schemes for MRC-D. Since they are special cases of the pDCF scheme, their achievable rates can be obtained by particularizing the achievable rate of pDCF for this setup.

\subsubsection{DF Scheme}
If we consider a pDCF scheme that does not perform any compression at the relay, i.e.,
$\hat{Y_R}=\emptyset$, we obtain the rate achievable by the pDF
scheme. Note that the optimal distribution of $X_R$ is given by  $p^*(x_r)$. Then, we have
\begin{IEEEeqnarray}{rCl}
R_{\mathrm{pDF}}&=& \min\{R_1,\sup_{p(x_1,u)}I(U;Y_R)\}.
\end{IEEEeqnarray}
From the Markov chain $U-X_1-Y_R$, we have  $I(U;Y_R)\leq
I(X_1;Y_R)$, where the equality is achieved by $U=X_1$. That is, the performance of  pDF is maximized by letting the relay decode the whole message. Therefore, the maximum rate achievable by pDF and DF for MRC-D coincide, and is given by
\begin{IEEEeqnarray}{rCl}
R_{\mathrm{DF}}=R_{\mathrm{pDF}}&=&\min\{R_1,\max_{p(x_1)} I(X_1;Y_R)\}.
\end{IEEEeqnarray}

We note that considering more advanced DF strategies based on list decoding as in \cite{Kim2007CodingPrimitiveRelay} does not increase the achievable rate in the MRC-D, since there is no direct link.

\subsubsection{CF Scheme} If the pDCF scheme does not perform any decoding at the relay, i.e., $U=V=\emptyset$, pDCF reduces to CF. Then, the achievable rate for the CF scheme in MRC-D is  given by
\begin{IEEEeqnarray}{rCl}\label{eq:CFRate}
R_{\mathrm{CF}}&=&\sup I(X_1;\hat{Y}_R|Z)\nonumber\\
&&\text{s.t. }\;R_1\geq I(\hat{Y}_R;Y_R|Z),\nonumber\\
&&\text{over }p(x_1)p(z)p(y_R|x_1,z)p(\hat{y}_R|y_R).
\end{IEEEeqnarray}

\subsection{Multihop Parallel Binary Symmetric Channel}\label{sec:Binary}

\begin{figure}[t!]
\centering
\includegraphics[width=0.5\textwidth]{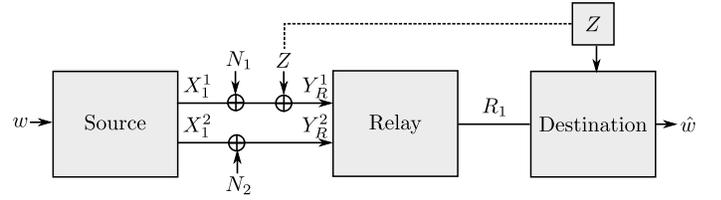}
\vspace{-3 mm} \caption{The parallel binary symmetric MRC-D with parallel source-relay links. The destination has side information about only one of the source-relay links.}
\label{fig:ParallMultiHop}
\end{figure}
 In this section we consider a special MRC-D as shown in Figure \ref{fig:ParallMultiHop}, which we call the \emph{parallel binary symmetric MRC-D}. For this setup, we characterize the optimal performance of the DF and CF schemes, and show that in general pDCF outperforms both, and that in some cases the cut-set bound is tight and coincides with the channel capacity. This example proves the suboptimality of both DF and CF  on their own for  ORC-D.

 In this scenario, the source-relay channel consists of two parallel binary symmetric channels. We have $X_{1}=(X_1^1,X_1^2)$, $Y_R=(Y_R^1,Y_R^2)$ and $p(y_R|x_R,z)=p(y_R^1|x^1_1,z)p(y_R^2|x_1^2)$ characterized by
\begin{IEEEeqnarray}{rCl}
Y_R^1 &=& X_1^1 \oplus N_1 \oplus Z,\quad\text{and } \quad
Y_R^2 = X_1^2 \oplus N_2,\nonumber
\end{IEEEeqnarray}
where $N_1$ and $N_2$  are i.i.d. Bernoulli random variables with $\text{Pr}\{N_1=1\}=\text{Pr}\{N_2=1\} = \delta$, i.e., $N_1\sim \text{Ber}(\delta)$ and $N_2\sim \text{Ber}(\delta)$. We consider a Bernoulli distributed state $Z$, $Z\sim \text{Ber}(p_z)$, which affects one of the two parallel channels, and is available at the destination. We have
$\mathcal{X}_1^1 = \mathcal{X}_1^2 = \mathcal{Y}_R^1 =  \mathcal{Y}_R^1 = \mathcal{N}_1 = \mathcal{N}_2 = \mathcal{Z}= \{0,1\}$.

From (\ref{eq:cutset}), the cut-set bound is given by
\begin{IEEEeqnarray}{rCl}\label{eq:BinCS}
R_{\mathrm{CS}} &=& \min\{R_1, \max_{p(x_1^1x_1^2)}I(X_1^1X_1^2;Y_R^1Y_R^2|Z)\}\nonumber\\
&=& \min\{R_1, 2(1- h_2(\delta))\},
\end{IEEEeqnarray}
where $h_2(\cdot)$ is the binary entropy function defined as $h_2(p) \triangleq
-p\log p -(1-p) \log(1-p)$.

The maximum DF rate is achieved by $X_1^1\sim\text{Ber}(1/2)$ and $X_1^2\sim\text{Ber}(1/2)$, and is found to be
\begin{IEEEeqnarray}{rCl}\label{eq:BinDF}
R_{\mathrm{DF}}  &=& \min\{R_1, \max_{p(x_1^1x_1^2)}I(X_1^1X_1^2;Y_R^1Y_R^2)\}\nonumber\\
&=&\min\{R_1, 2- h_2(\delta\star p_z)- h_2(\delta)\},
\end{IEEEeqnarray}
where $\alpha\star\beta\triangleq \alpha(1-\beta)+(1-\alpha)\beta$.

Following (\ref{eq:CFRate}), the rate achievable by the CF scheme in the parallel binary symmetric MRC-D is given by
\begin{IEEEeqnarray}{rCl}\label{eq:CFparallelRate}
R_{\mathrm{CF}} &=& \max I(X_1^1X_1^2,\hat{Y}_R|Z), \\
&&\text{s.t. } \;\; R_1 \geq I(Y_R^1Y_R^2; \hat{Y}_R|Z)\nonumber\\
&& \text{over } p(z)p(x_1^1x_1^2)p(y_R^1|z,x_1^1)p(y_R^2|x_2)p(\hat{y}_R|y_R^1y_R^2).\nonumber
\end{IEEEeqnarray}

Let us define $h_2^{-1}(q)$ as the inverse of the entropy function $h_2(p)$ for $q\geq0$. For $q<0$, we define $h_2^{-1}(q)=0$.

As we show in the next lemma, the achievable CF rate in (\ref{eq:CFparallelRate}) is maximized by transmitting independent channel inputs over the two parallel links to the relay by setting $X_1^1\sim\text{Ber}(1/2)$, $X_1^2\sim\text{Ber}(1/2)$, and by independently compressing each of the channel outputs $Y_R^1$ and $Y_R^2$ as $\hat{Y}_R^1=Y_R^1\oplus Q_1$ and $\hat{Y}_R^2=Y_R^2\oplus Q_2$, respectively, where $Q_1\sim\text{Ber}(h_2^{-1}(1-R_1/2))$ and $Q_2\sim\text{Ber}(h_2^{-1}(1-R_1/2))$. Note that for $R_1\geq 2$, the channel outputs can be compressed errorlessly. The maximum achievable CF rate is given in the following lemma.

\begin{lemma}\label{lemm:CF}
The maximum rate achievable by CF over the parallel binary symmetric MRC-D is given by
\begin{IEEEeqnarray}{rCl}\label{eq:CFParBin}
R_{\mathrm{CF}}=2\left(1-h_2\left(\delta\star h_2^{-1}\left(1-\frac{R_1}{2}\right)\right)\right).
\end{IEEEeqnarray}
\end{lemma}
\begin{proof}
See Appendix \ref{app:CF}.
\end{proof}

Now, we consider the pDCF scheme for the parallel binary symmetric MRC-D. Although we have not been able to characterize the optimal choice of $(U,\hat{Y}_R,X_1^1,X_1^2)$ in general, we provide an achievable scheme that outperforms both DF and CF schemes and meets the cut-set bound in some regimes. Let $X_1^1\sim \text{Ber}(1/2)$ and $X_1^2\sim \text{Ber}(1/2)$ and  $U=X_1^2$, i.e., the relay decodes the channel input $X_1^2$, while $Y_R^1$ is compressed using $\hat{Y}_R=Y_R^1+Q$, where $Q\sim\text{Ber}(h^{-1}_2(2-h_2(\delta)-R_1))$. The rate achievable by this scheme is given in the following lemma.
\begin{lemma}\label{lemm:pDCFRate}
A lower bound on the  achievable pDCF rate over the parallel binary symmetric MRC-D is given by
\begin{IEEEeqnarray}{lCl}
R_{\mathrm{pDCF}} \nonumber
\geq \min\{R_1,2-h_2(\delta)\!-\!h_2\left(\delta\star h_2^{-1}\left(2-h_{2}(\delta)\!-\!R_1\right)\right)\}\nonumber.
\end{IEEEeqnarray}
\end{lemma}
\begin{proof}
See Appendix \ref{app:pDCFRate}.
\end{proof}
\begin{figure}[t!]
\centering
\includegraphics[width=0.5\textwidth]{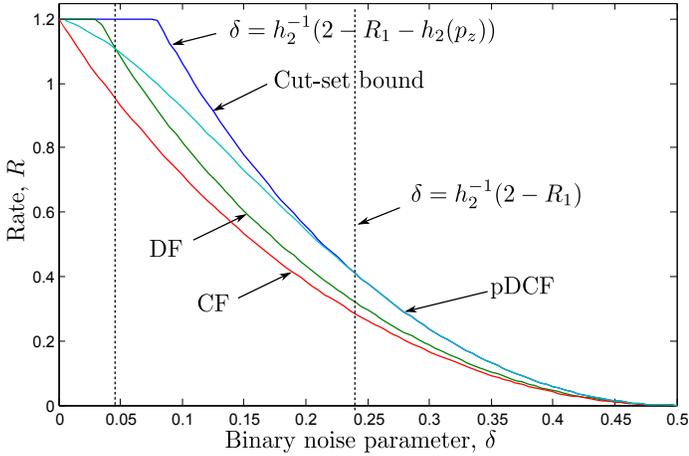}
\vspace{-1 mm} \caption{Achievable rates and the cut-set upper bound for the  parallel binary symmetric MRC-D with respect to the binary noise parameter $\delta$, for $R_1=1.2$ and $p_z=0.15$.}
\label{fig:BinaryOpt}
\end{figure}

We notice that for $p_z\leq h_2^{-1}\left(2-h_{2}(\delta)-R_1\right)$, or equivalently, $\delta\leq h_2^{-1}\left(2-h_{2}(p_z)-R_1\right)$, the proposed pDCF is outperformed by DF. In this regime, pDCF can achieve the same performance by decoding both channel inputs, reducing to DF.

Comparing the cut-set bound expression in (\ref{eq:BinCS}) with $R_{\mathrm{DF}}$ in (\ref{eq:BinDF}) and $R_{\mathrm{CF}}$ in (\ref{eq:CFParBin}),
we observe that DF achieves the cut-set bound if $R_1\leq 2-h(\delta\star p_z)-h(\delta)$ while $R_{\mathrm{CF}}$ coincides with the cut-set bound if $R_1\geq2$. On the other hand, the proposed suboptimal pDCF scheme achieves the cut-set bound if $R_1\geq 2-h_2(\delta)$, i.e., for $\delta\geq h^{-1}_2(2-R_1)$. Hence, the capacity of the parallel binary symmetric MRC-D in this regime is achieved by pDCF, while both DF and CF are suboptimal, as stated in the next lemma.
\begin{lemma}
If $R_1<2$ and $\delta\geq h^{-1}_2(2-R_1)$, pDCF achieves the capacity of the parallel binary symmetric MRC-D, while pure CF and DF are both suboptimal under these constraints. For $R_1\geq2$, both  CF and pDCF achieve the capacity.
\end{lemma}

The achievable rates of DF, CF and pDCF, together with the cut-set bound are shown in Figure \ref{fig:BinaryOpt} with respect to $\delta$ for $R_1=1.2$ and $p_z=0.15$. We observe that in this setup, DF outperforms CF in general, while for $\delta\leq h^{-1}_2(2-R_1-h_2(p_z))=0.0463$, DF outperforms the proposed suboptimal pDCF scheme as well. We also observe that pDCF meets the cut-set bound for $\delta\geq h^{-1}_2(2-R_1)=0.2430$, characterizing the capacity in this regime, and proving the suboptimality of both the DF and CF schemes when they are used on their own.

\subsection{Multihop Binary Symmetric Channel}\label{sec:binarySingle}

In order to gain further insights into the proposed pDCF scheme, we look into the \emph{binary symmetric} MRC-D, in which, there is only a single channel connecting the source to the relay, given by
\begin{IEEEeqnarray}{rCl}\label{eq:singleBinary}
Y_R=X_1\oplus N\oplus Z,
\end{IEEEeqnarray}
where $N\sim\text{Ber}(\delta)$ and $Z\sim\text{Ber}(p_z)$.

Similarly to Section \ref{sec:Binary}, the cut-set bound and the maximum achievable rates for DF and CF are found as
\begin{IEEEeqnarray}{rCl}
R_{\mathrm{CS}}&=& \min\{R_1, 1- h_2(\delta)\},\label{eq:CSbin}\\
R_{\mathrm{DF}}&=& \min\{R_1, 1- h_2(\delta\star p_z)\},\\
R_{\mathrm{CF}}&=& 1-h_2(\delta\star h_2^{-1}(1-R_1))),
\end{IEEEeqnarray}
where $R_{\mathrm{DF}}$ is achieved by $X_1\sim\text{Ber}(1/2)$, and $R_{\mathrm{CF}}$ can be shown  to be maximized by $X_1\sim\text{Ber}(1/2)$ and $\hat{Y}_R=Y_R\oplus Q$, where $Q\sim\text{Ber}(h_2^{-1}(1-R_1))$ similarly to Lemma \ref{lemm:CF}. Note that, for $Y_R$ independent of $Z$, i.e., $p_z=0$, DF achieves the cut-set bound while CF is suboptimal. However, CF outperforms DF whenever $p_z\geq h_2^{-1}(1-R_1)$.

 For the pDCF scheme, we consider binary $(U,X_1,\hat{Y}_R)$, with $U\sim\text{Ber}(p)$, a superposition codebook $X_1=U\oplus W$, where $W\sim\text{Ber}(q)$, and $\hat{Y}_R=Y_R\oplus Q$ with $Q\sim\text{Ber}(\alpha)$. As stated in the next lemma, the maximum achievable rate of this pDCF scheme is obtained by reducing it to either DF or CF, depending on the values of $p_z$ and $R_1$.
 \begin{lemma}
 For the binary symmetric MRC-D, pDCF with binary $(U,X_1,\hat{Y}_R)$ achieves the following rate.
\begin{IEEEeqnarray}{rCl}
R_{\mathrm{pDCF}}&=& \max\{R_{\mathrm{DF}},R_{\mathrm{CF}}\}\\
&=&\!\begin{cases}
 \min\{R_1, 1- h_2(\delta\star p_z)\}&\text{if } p_z< h_2^{-1}(1-R_1),\\
1-h_2(\delta\star h_2^{-1}(1-R_1))&\text{if
} p_z\geq h_2^{-1}(1-R_1).\nonumber
 \end{cases}
\end{IEEEeqnarray}
\end{lemma}

This result justifies the pDCF scheme proposed in Section \ref{sec:Binary} for the
  parallel binary symmetric MRC-D. Since the channel $p(y_1^2|x_2)$ is independent of the channel state $Z$, the largest rate is are achieved if the relay decodes $X_1^2$ from $Y_R^2$. However, for channel $p(y_1^1|x_1,z)$, which depends on $Z$, the relay either decodes $X_1^1$, or compress $Y_R^1$, depending on $p_z$.

\subsection{Multihop Gaussian Channel with State}\label{sec:GaussEx}
Next, we consider an AWGN multihop channel, called \emph{Gaussian} MRC-D, in which the source-relay link
is characterized by $Y_R=X_1+V$, while the destination has access to correlated state information $Z$. We assume that $V$ and $Z$ are zero-mean
jointly Gaussian random variables with a covariance matrix
\begin{IEEEeqnarray}{rCl}\label{eq:cov_noise}
\mathbf{C}_{ZV}=\left[\begin{matrix}1&\rho\\ \rho &
1\end{matrix}\right].
\end{IEEEeqnarray}
The channel input
at the source has to satisfy the power constraint $E[|X_1^n|^2]\leq
nP$. Finally, the relay and the destination are connected by a
noiseless link of rate $R_1$ (see Figure
\ref{fig:SingleGaussianMultiHop} for the channel model).

\begin{figure}[t!]
\centering
\includegraphics[width=0.5\textwidth]{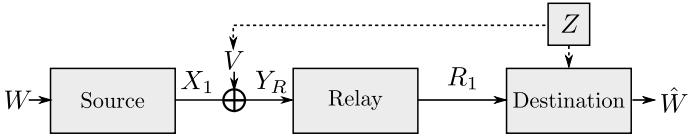}
\vspace{-3 mm} \caption{The multihop Gaussian relay channel with
source-relay channel state information available at the
destination.} \label{fig:SingleGaussianMultiHop}
\end{figure}

In this case, the cut-set bound is given by
\begin{IEEEeqnarray}{rCl}
R_{\mathrm{CS}}=\min\left\{R_1,\frac{1}{2}\log\left(1+\frac{P}{1-\rho^2}\right)\right\}.
\end{IEEEeqnarray}
It easy to
characterize the optimal DF rate, achieved by a Gaussian input, as follows:
\begin{IEEEeqnarray}{rCl}
R_{\mathrm{DF}}=\min\left\{R_1,\frac{1}{2}\log(1+P)\right\}.
\end{IEEEeqnarray}

For CF and pDCF, we consider the achievable rate when the random
variables $(X_1,U,\hat{Y}_R)$ are constrained to be jointly
Gaussian, which is a common assumption in evaluating achievable
rates, yet potentially suboptimal. For CF, we generate the
compression codebook using $\hat{Y}_R=Y_R+Q$, where
$Q\sim\mathcal{N}(0,\sigma^2_{q})$. Optimizing over $\sigma^2_q$,
the maximum achievable rate is given by
\begin{IEEEeqnarray}{rCl}
R_{\mathrm{CF}}&=&R_1-\frac{1}{2}\log\left(\frac{P+2^{2R_1}(1-\rho^2)}{P+1-\rho^2}\right).
\end{IEEEeqnarray}

For pDCF, we let $U\sim\mathcal{N}(0,\alpha P_1)$, and $X_{1} = U+T$ to
be a superposition codebook where $T$ is independent of $U$ and
distributed as $T\sim\mathcal{N}(0,\bar{\alpha} P_1)$, where
$\bar{\alpha}\triangleq 1-\alpha$. We generate a quantization
codebook using the test channel $\hat{Y}_{R}=Y_R+Q$ as in CF. Next
lemma shows that with this choice of random variables, pDCF reduces
either to pure DF or pure CF, similarly to the multihop binary model in Section \ref{sec:binarySingle}.
\begin{lemma}\label{lemm:pDCFGauss}
 The optimal achievable rate for pDCF with jointly Gaussian $(X_1,U,\hat{Y}_R)$ is given by
\begin{IEEEeqnarray}{lCl}
R_{\mathrm{pDCF}}=\max\{R_{\mathrm{DF}},R_{\mathrm{CF}}\}\\
 =\begin{cases}
\min\left\{R_1,1/2\log(1+P)\right\}&\text{if } \rho^2\leq 2^{-2R_1}(1+P),\\
R_1-\frac{1}{2}\log\left(\frac{P+2^{2R_1}(1-\rho^2)}{P+1-\rho^2}\right)&\text{if
} \rho^2>2^{-2R_1}(1+P).\nonumber
 \end{cases}
\end{IEEEeqnarray}
\end{lemma}

\begin{proof}
See Appendix \ref{app:pDCFGauss}.
\end{proof}
\begin{figure}[t!]
\centering
\includegraphics[width=0.5\textwidth]{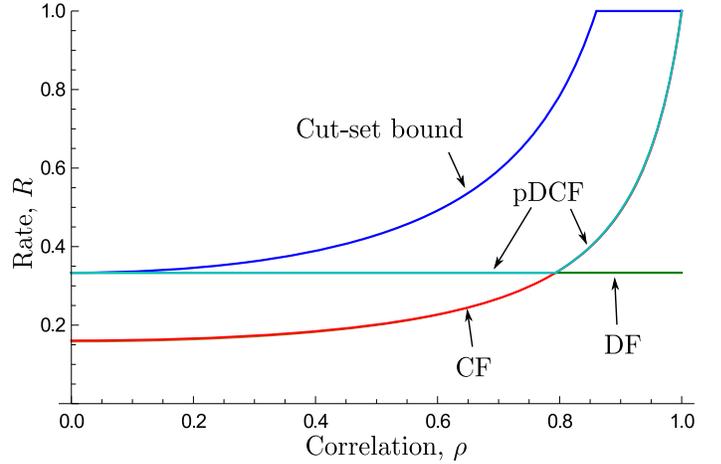}
\vspace{-3 mm} \caption{Achievable rates and the cut-set upper bound
for the multihop AWGN relay channel with source-relay channel state
information at the destination for $R_1=1$ and $P=0.3$.}
\label{fig:GaussianTwoHop} \vspace{-3 mm}
\end{figure}

In Figure \ref{fig:GaussianTwoHop} the achievable rates are compared
with the cut-set bound. It is shown that DF achieves the best rate
when the correlation coefficient $\rho$ is low, i.e., when the
destination has low quality channel state information, while CF
achieves higher rates for higher values of $\rho$. It is seen that
pDCF achieves the best of the two transmission schemes. Note also
that for $\rho=0$ DF meets the cut-set bound, while for $\rho=1$
CF meets the cut-set bound.

Although this example proves the
suboptimality of the DF scheme for the channel model under
consideration, it does not necessarily lead to the suboptimality of the CF scheme as we have constrained the auxiliary random variables to be Gaussian.

\section{Comparison with the Cut-Set Bound}\label{seq:cutset}
In the examples considered in Section \ref{sec:Examples}, we have
seen that for certain conditions, the choice of certain random
variables allows us to show that the cut-set bound and the capacity
coincide. For example, we have seen that for the parallel binary symmetric MRC-D the proposed pDCF scheme achieves the cut-set bound for $\delta\geq h^{-1}_2(2-R_1)$, or  Gaussian random variables meet the cut-set bound for $\rho=0$ or $\rho=1$ in the Gaussian MRC-D. An interesting question is whether the capacity expression in Theorem \ref{th:capreg} always coincides with the cut-set bound or not; that is, whether the cut-set bound is tight for the relay channel model under consideration.

To address this question, we consider the multihop binary channel in (\ref{eq:singleBinary}) for $Z\sim\text{Ber}(1/2)$. The  capacity $\mathcal{C}$ of this channel is given in the following lemma.
\begin{lemma}\label{lemm:CutCap}
The capacity of the binary symmetric MRC-D with $Y_R=X_1\oplus N\oplus Z$, where $N\sim\text{Ber}(\delta)$ and $Z\sim\text{Ber}(1/2)$, is achieved by CF and pDCF, and is given by
\begin{IEEEeqnarray}{rCl}
\mathcal{C}=1-h_2(\delta\star h_2^{-1}(1-R_1)).
\end{IEEEeqnarray}
\end{lemma}
\begin{proof}
See Appendix \ref{app:CutCap}.
\end{proof}

From (\ref{eq:CSbin}), the cut-set bound is given by $R_{\mathrm{CS}}=1-h_2(\delta)$. It then follows that in general the capacity is below the cut-set bound. Note that for this setup, $R_{\mathrm{DF}}=0$  and pDCF reduces to CF, i.e., $R_{\mathrm{pDCF}}=R_{\mathrm{CF}}$. See Figure \ref{fig:BinaryCS} for comparison of the capacity with the cut-set bound for varying $\delta$ values.

 CF suffices to achieve the capacity of the binary symmetric MRC-D for $Z\sim\text{Ber}(1/2)$. While in general pDCF outperforms DF and CF, in certain cases these two schemes are sufficient to achieve the cut-set bound, and hence, the capacity. For the ORC-D model introduced in Section \ref{sec:Model}, the cut-set bound is given by
\begin{IEEEeqnarray}{rCl}
R_{\mathrm{CS}}=R_2+\min\{R_1,\max_{p(x_1)}I(X_1;Y_R|Z)\}.
\end{IEEEeqnarray}

Next, we present four cases for which the cut-set bound is
achievable, and hence, is the capacity:

\begin{itemize}
\item[]\hspace{-7mm} \textit{Case 1)} If $I(Z;Y_R)=0$, the setup reduces to the class of
orthogonal relay channels studied in \cite{Gamal2005OrthRelay}, for
which the capacity is known to be achieved by pDF.
\item[]\hspace{-7mm} \textit{Case 2)} If $H(Y_R|X_1Z)=0$, i.e., $Y_R$ is a deterministic function of $X_1$ and $Z$, the capacity, given by \[R_2+\min\{R_1,\max_{p(x_1)}I(X_1;Y_R|Z)\},\] is achievable by CF.
\item[]\hspace{-7mm}\textit{ Case 3)} If $\max_{p(x_1)}I(X_1;Y_R)\geq R_1$, the capacity, given by $\mathcal{C}=R_2+R_1$, is achievable by pDF.
\item[]\hspace{-7mm}\textit{ Case 4)} Let $\arg \max_{p(x_1)}I(X_1;Y_R|Z)=\bar{p}(x_1)$. If $R_1>H(\bar{Y}_R|Z)$ for $\bar{Y}_R$ induced by $\bar{p}(x_1)$, the capacity, given by $R_2+I(\bar{X}_1;\bar{Y}_R|Z)$, is achievable by CF.
\end{itemize}
\begin{proof}
See Appendix \ref{app:CutSetProof}.
\end{proof}

\begin{figure}[t!]
\centering
\includegraphics[width=0.5\textwidth]{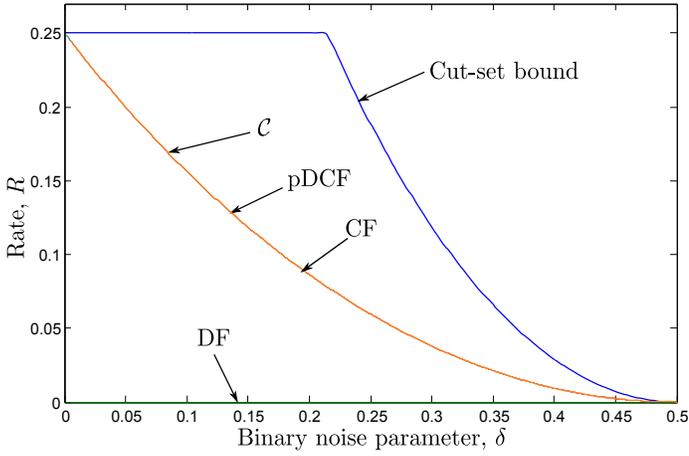}
\vspace{-1 mm} \caption{Achievable rates, capacity and cut-set upper bound for the multihop binary relay channel with respect to $\delta$ for $R_1=0.25$ and $p_z=0.5$.}
\label{fig:BinaryCS}
\end{figure}

These cases can be observed in the examples from Section \ref{sec:Examples}. For example, in the Gaussian MRC-D with $\rho=0$, $Y_R$ is independent of $Z$, and thus, DF meets the cut-set bound as stated in Case 1. Similarly, for $\rho=1$ CF meets the cut-set bound since $Y_R$ is a deterministic function of $X_R$ and $Z$, which corresponds to Case 2.

For the parallel binary symmetric MRC-D in Section \ref{sec:Binary}, pDCF achieves the cut-set bound if $\delta\geq h^{-1}_2(2-R_1)$ due to the following reasoning. Since $Y_R^1$ is independent of $X_1^1$, from Case 1, DF should achieve the cut-set bound. Once $X_1^1$ is decoded, the available rate to compress $Y_2$ is given by $R_1-I(X_1;Y_1)=R_1-1+h_2(\delta)$, and the entropy of $Y_2$ conditioned on the channel state at the destination is given by $H(Y_2|Z)=1-h_2(\delta)$. For $\delta\geq h^{-1}_2(2-R_1)$ we have $R_1-I(X_1;Y_1)\geq H(Y_2|Z)$. Therefore the relay can compress $Y_2$ losslessly, and transmit to the destination. This corresponds to Case 4. Thus, the capacity characterization in the parallel binary symmetric MRC-D is due to a combination of Case 1 and Case 4.

\section{Conclusion}
We have considered a class of orthogonal relay channels,
in which the channels connecting the source to the relay and the destination, and the relay to the destination, depend on a state sequence, known at the destination. We have
characterized the capacity of this class of relay channels, and
shown that it is achieved by the partial decode-compress-and-forward
(pDCF) scheme. This is the first three-terminal relay channel model
for which the pDCF is shown to be capacity achieving while partial
decode-and-forward (pDF) and compress-and-forward (CF) schemes are
both suboptimal in general. We have also shown that, in general, the capacity of this channel is below the cut-set bound.


\begin{appendices}
\section{Proof of Lemma \ref{lem:Achi}} \label{app:ProofLemmaAch}
In the rate expression and joint pmf in Theorem \ref{th:pDCF}, we set $X^n=(X_1^n,X_2^n)$, $Y^n=(Y_{1}^n,Y_{2}^n,Z^n)$,
$V=\emptyset$, and generate $X_R^n$ and $X_2^n$ independent of the  rest of the random
 variables with distributions $p^*(x_R)$ and $p^*(x_2)$, which maximize the mutual information terms in (\ref{eq:R0R1rates}), respectively. Under this set of distributions we have
\begin{IEEEeqnarray}{rCl}
I(X;Y\hat{Y}_R|X_RU)&=&I(X_1 X_2;Y_{1}Y_{2}\hat{Y}_RZ|X_R,U)\nonumber\\
&\stackrel{(a)}{=}&I(X_1 X_2;Y_{2}\hat{Y}_R|X_RUZ)\nonumber\\
&\stackrel{(b)}{=}&I(X_2;Y_2|Z)+I(X_1;\hat{Y}_R|UZ)\nonumber\\
&=&R_2+I(X_1;\hat{Y}_R|UZ),\nonumber\\
I(U;Y_R|X_RV)&=&I(U;Y_R|X_R)\stackrel{(c)}{=}I(U;Y_R),\nonumber\\
I(XX_R;Y)&=&I(X_1X_2X_R;Y_{1}Y_{2}Z)\nonumber\\
&\stackrel{(d)}{=}&I(X_2X_R;Y_{1}Y_{2}|Z)\nonumber\\
&\stackrel{(e)}{=}&I(X_R;Y_1)+I(X_2;Y_{2}|Z)=R_1+R_2,\nonumber\\
I(\hat{Y}_R;Y_R|XX_RUY)&=&I(\hat{Y}_R;Y_R|X_RX_1X_2UY_{1}Y_{2}Z)\nonumber\\
&\stackrel{(f)}{=}&I(\hat{Y}_R;Y_R|X_RX_1X_2UY_{2}Z)\nonumber\\
&\stackrel{(g)}{=}&I(\hat{Y}_R;Y_R|X_1UZ),\nonumber\\
I(\hat{Y}_R;Y_R|YX_RU)&=&I(\hat{Y}_R;Y_R|Y_{1}Y_{2}ZX_RU)\nonumber\\
&\stackrel{(h)}{=}&I(\hat{Y}_R;Y_R|UZ),\nonumber\\
I(X_R;Y|V)&=&I(X_R;Y_{1}Y_{2}Z)=I(X_R;Y_{1})=R_1,\nonumber
\end{IEEEeqnarray}
where $(a)$ is due to the Markov chain $(X_1X_2)-X_R-Y_{1}$;
$(b),(c),(e),(f),(g),(h)$ are due to the independence of $(U,X_1)$
and $X_R$, and $(d)$ is due to the Markov chain
$(Y_{1}Y_{2})-(X_2X_RZ)-X_1$.

Then, (\ref{eq:GenPartialDFCF}) reduces to the following rate
\begin{IEEEeqnarray}{rCl}
R=&\underset{\mathcal{P}}{\sup}&\,\min\{I(U;Y_R)+R_2+I(X_1;\hat{Y}_R|UZ),\\
&&\;\qquad R_2+R_1\nonumber-I(\hat{Y}_R;Y_R|X_1UZ)\}, \nonumber\\
&\text{s.t. }&R_1\geq I(\hat{Y}_R;Y_R|UZ).
\end{IEEEeqnarray}
Focusing on the joint distributions in $\mathcal{P}$ such that the minimum in $R$ is achieved for the first argument, i.e.,
\begin{IEEEeqnarray}{rCl}
R_1-I(\hat{Y}_R;Y_R|X_1UZ) \geq I(U;Y_R)+I(X_1;\hat{Y}_R|UZ),\nonumber
\end{IEEEeqnarray}
and  using the chain rule for the mutual information, the rate achievable by pDCF is lower bounded by
\begin{IEEEeqnarray}{rCl}
R\geq&\underset{\mathcal{P}}{\sup}\;&R_2+I(U;Y_R)+I(X_1;\hat{Y}_R|UZ)\nonumber\\
&\text{s.t. }& R_1\geq I(U;Y_R)+I(X_1Y_R;\hat{Y}_R|UZ),\label{eq:barrcond1}\\
&&R_1\geq I(\hat{Y}_R;Y_R|UZ)\label{eq:barrcond2}.
\end{IEEEeqnarray}
 From
(\ref{eq:barrcond1}), we have
\begin{IEEEeqnarray}{rCl}\label{eq:Eqcond}
R_1&\geq&I(U;Y_R)+I(X_1Y_R;\hat{Y}_R|UZ)\nonumber\\
&\stackrel{(a)}{=}&I(U;Y_R)+I(\hat{Y}_R;Y_R|UZ)\nonumber\\
&\geq&I(\hat{Y}_R;Y_R|UZ),
\end{IEEEeqnarray}
where $(a)$ is due to the Markov chain $\hat{Y}_R-(UY_R)-(X_1Z)$. Hence, (\ref{eq:barrcond1}) implies
(\ref{eq:barrcond2}), i.e., the latter condition is redundant, and $R\geq \mathcal{C}$. Therefore the capacity expression $\mathcal{C}$ in (\ref{eq:capreg}) is achievable by pDCF. This concludes the proof.

\section{Proof of Lemma \ref{lem:Conv1}}\label{app:Converse1}

Consider any sequence of $(2^{nR},n,\nu_n)$ codes such that
$\lim_{n\rightarrow\infty}\nu_n\rightarrow 0$. We need to show that
$R\leq R_{up}$.

Let us define $U_i\triangleq(Y_{R1}^{i-1}X_{1i+1}^nZ^{n\setminus
i})$ and $\hat{Y}_{Ri}\triangleq(Y_{1i+1}^{n})$. For such
$\hat{Y}_{Ri}$ and $U_i$, the following Markov chain holds
\begin{IEEEeqnarray}{rCl}\label{eq:Uchain}
\hat{Y}_{Ri}-(U_i,Y_{Ri})-(X_{1i},X_{2i},Z_{i},Y_{1i},Y_{2i},X_{Ri}).
\end{IEEEeqnarray}

From Fano's inequality, we have
\begin{IEEEeqnarray}{rCl}\label{eq:Fano}
H(W|Y_{1}^nY_{2}^nZ^n)\leq n\epsilon_n,
\end{IEEEeqnarray}
such that $\epsilon_n\rightarrow 0$ as $n\rightarrow \infty$.

First, we derive the following set of inequalities related to the capacity of the source-destination channel.
\begin{IEEEeqnarray}{rCl}\label{eq:ncutset}
nR&=&H(W)\nonumber\\
&\stackrel{(a)}{=}&I(W;Y_1^nY_2^n|Z^n)+H(W|Y_1^nY_2^nZ^n)\nonumber\\
&\stackrel{(b)}{\leq} &I(X_1^nX_2^n;Y_1^nY_2^n|Z^n)+n\epsilon_n,
\end{IEEEeqnarray}
where $(a)$ follows from the independence of $Z^n$ and $W$; and $(b)$ follows from Fano's inequality in (\ref{eq:Fano}).

We also have the following inequalities:
\begin{IEEEeqnarray}{rCl}\label{eq:R1Inequ}
I(X_2^n;Y_2^n|Z^n)\!&=&\!\sum_{i=1}^{n}H(Y_{2i}|Z^nY_{21}^{i-1})-H(Y_{2i}|Z^nY_{21}^{i-1}X_2^n)\nonumber\\
&\stackrel{(a)}{\leq}&\sum_{i=1}^{n}H(Y_{2i}|Z_i)-H(Y_{2i}|Z_iX_{2i})\nonumber\\
&=&\sum_{i=1}^{n}I(X_{2i};Y_{2i}|Z_i)\nonumber\\
&\stackrel{(b)}{=}&nI(X_{2Q'};Y_{2Q'}|Z_{Q'}Q')\nonumber\\
&\stackrel{(c)}{\leq} &nI(X_{2Q'};Y_{2Q'}|Z_{Q'})\nonumber\\
&\stackrel{(d)}{\leq}& nR_2,
\end{IEEEeqnarray}
where $(a)$ follows since conditioning reduces entropy, $(b)$ follows by defining $Q'$ as a uniformly distributed random variable over $\{1,...,n\}$ and $(X_{2Q'},Y_{2Q'},Z_{Q'})$ as a pair of random variables satisfying $\text{Pr}\{X_{2i}=x_2,Y_{2i}=y_2,Z_{i}=z\}=\text{Pr}\{X_{2Q'}=x_2,Y_{2Q'}=y_2,,Z_{Q'}=z|Q=i\}$ for $i=1,...,n$, $(c)$ follows from the Markov chain relation $Q'-X_{2Q'}-Y_{2Q'}$ and $(d)$ follows from the definition of $R_2$ in (\ref{eq:R0R1rates}). Following the same steps, we obtain
\begin{IEEEeqnarray}{rCl}\label{eq:R0Inequ}
I(X_R^n;Y_1^n|Z^n)\leq nR_1.
\end{IEEEeqnarray}
\interdisplaylinepenalty=1000
 Then, we can bound the achievable rate as,
\begin{IEEEeqnarray}{rCl}\label{eq:cond3}
 nR
 &=&I(W;Y_{1}^nY_{2}^nZ^n)+H(W|Y_{1}^nY_{2}^nZ^n)\nonumber\\
 &\stackrel{(a)}{\leq}& I(W;Y_{1}^nY_{2}^nZ^n)+n\epsilon_n\nonumber\\
 &\stackrel{(b)}{=}&I(W;Y_2^n|Z^n)+I(W;Y_1^n|Y_2^nZ^n)+n\epsilon_n\nonumber\\
&\stackrel{(c)}{\leq}&I(X_2^n;Y_2^n|Z^n)+I(W;Y_1^n|Y_2^nZ^n)+n\epsilon_n\nonumber\\
  &\stackrel{(d)}{\leq}&nR_2+H(Y_1^n|Y_2^nZ^n)-H(Y_1^n|WZ^n)+n\epsilon_n\nonumber\\
 &\stackrel{(e)}{\leq}&nR_2+H(Y_1^n|Z^n)-H(Y_1^n|WX_1^nZ^n)+n\epsilon_n\nonumber\\
 &\stackrel{(f)}{=}&nR_2+I(X_1^n;Y_1^n|Z^n)+n\epsilon_n\nonumber\\
 &\stackrel{(g)}{\leq}&nR_2+H(X_1^n)-H(X_1^n|Y_1^nZ^n)+n\epsilon_n\nonumber\\
 &=&nR_2+\sum_{i=1}^nH(X_{1i}|X_{1i+1}^n)-H(X_1^n|Y_1^nZ^n)+n\epsilon_n\nonumber\\
 &\stackrel{(h)}{\leq}&nR_2+\sum_{i=1}^n\bigg[I(Y_{R1}^{i-1}Z^{n\setminus i};Y_{Ri})\nonumber\\
 &&+H(X_{1i}|X_{1i+1}^n)\bigg]-H(X_1^n|Y_1^nZ^n)+n\epsilon_n\nonumber\\
 &=&nR_2+\sum_{i=1}^n\bigg[I(Y_{R1}^{i-1}Z^{n\setminus i}X_{1i+1}^n;Y_{Ri})\nonumber\\
 &&-I(X_{1i+1}^n;Y_{Ri}|Y_{R1}^{i-1}Z^{n\setminus i})\nonumber+H(X_{1i}|Y_{R1}^{i-1}Z^{n\setminus i}X_{1i+1}^n)\nonumber\\
 &&+I(X_{1i};Y_{R1}^{i-1}Z^{n\setminus i}|X_{1i+1}^n)\bigg]-H(X_1^n|Y_1^nZ^n)+n\epsilon_n\nonumber\\
 &\stackrel{(i)}{=}&nR_2+\sum_{i=1}^n\bigg[I(Y_{R1}^{i-1}Z^{n\setminus i}X_{1i+1}^n;Y_{Ri})\nonumber\\
 && +H(X_{1i}|Y_{R1}^{i-1}Z^{n\setminus i}X_{1i+1}^n)\bigg]-H(X_1^n|Y_1^nZ^n)+n\epsilon_n\nonumber\\
 &=&nR_2+\sum_{i=1}^n\bigg[I(U_i;Y_{Ri})+H(X_{1i}|U_i)\bigg]\nonumber\\
 &&-H(X_1^n|Y_1^nZ^n)+n\epsilon_n\nonumber\nonumber\\
&\stackrel{(j)}{\leq}& nR_2+\sum_{i=1}^n\bigg[I(U_i;Y_{Ri})+H(X_{1i}|U_i)\nonumber\\
&&-H(X_{1i}|U_iZ_i\hat{Y}_{Ri})\bigg]+n\epsilon_n\nonumber\\
&\stackrel{(k)}{=}&nR_2+\sum_{i=1}^n\bigg[I(U_i;Y_{Ri})+I(X_{1i};\hat{Y}_{Ri}|U_iZ_i)\bigg]+n\epsilon_n\nonumber,
\end{IEEEeqnarray}
\interdisplaylinepenalty=10000
where $(a)$ is due to Fano's inequality; $(b)$ is due to the chain
rule and the independence of $Z^n$ from $W$; $(c)$ is due to the
data processing inequality; $(d)$ is due to the Markov chain relation $Y_2^n-(W,Z^n)-Y_1^n$ and (\ref{eq:R1Inequ}); $(e)$ is due to the fact
that conditioning reduces entropy, and that $X_1^n$ is a deterministic function of
$W$; $(f)$ is due to the Markov chain relation $Y^n_1-X_1^n-W$; $(g)$ is due
to the independence of $Z^n$ and $X_1^n$; $(i)$ follows because
\begin{IEEEeqnarray}{rCl}
\sum_{i=1}^nI(X_{1i};Y_{R1}^{i-1}Z^{n\setminus i}|X_{1i+1}^n)
&\stackrel{(l)}{=}&\sum_{i=1}^nI(X_{1i};Y_{R1}^{i-1}|X_{1i+1}^nZ^{n\setminus i})\nonumber\\
&\stackrel{(m)}{=}&\sum_{i=1}^nI(X_{1i+1}^n;Y_{Ri}|Y_{R1}^{i-1}Z^{n\setminus
i}),\nonumber
\end{IEEEeqnarray}
where $(l)$ is due to the independence of $Z^n$ and $X_1^n$; and
$(m)$ is the conditional version of Csisz\'ar's equality
\cite{elGamal:book}. Inequality $(j)$ is due to the following bound,
\interdisplaylinepenalty=1000
\begin{IEEEeqnarray}{rCl}\label{eq:boundIxdyd}
H(X_1^n|Y_1^nZ^n)
&=&\sum_{i=1}^nH(X_{1i}|X_{1i+1}^nZ^nY_1^n)\nonumber\\
&\geq&\sum_{i=1}^nH(X_{1i}|Y_{R1}^{i-1}X_{1i+1}^nZ^nY_1^n)\nonumber\\
&\stackrel{(n)}{=}&\sum_{i=1}^nH(X_{1i}|Y_{R1}^{i-1}X_{1i+1}^nZ^nY_{1i+1}^n)\nonumber\\
&=&\sum_{i=1}^nH(X_{1i}|U_iZ_i\hat{Y}_{Ri}),
\end{IEEEeqnarray}
where $(n)$ is follows from the Markov chain relation
$X_{1i}-(Y_{R1}^{i-1}X_{1i+1}^nZ^nY_{1i+1}^n)-Y_{11}^{i}$, and
noticing that $X_{Ri}=f_{r,i}(Y_{R1}^{i-1})$. Finally, $(k)$ is due to the
fact that $Z_i$ independent of $(X_{1i},U_{i})$.
\interdisplaylinepenalty=10000

We can also obtain the following sequence of inequalities
\begin{IEEEeqnarray}{rCl}\label{eq:cond1a}
&nR_1&+nR_2\nonumber\\
&\stackrel{(a)}{\geq}& I(X_R^n;Y_{1}^n|Z^n)+I(X_{2}^n;Y_{2}^n|Z^n)\nonumber\\
&\stackrel{(b)}{\geq}&H(Y_{2}^n|Z^n)-H(Y_{2}^n|X_2^nZ^n)\nonumber\\
&&+H(Y_{1}^n|Y_{2}^nZ^n)-H(Y_{1}^n|X_R^nZ^n)\nonumber\\
&=&H(Y_{1}^nY_{2}^n|Z^n)-H(Y_{2}^n|X_2^nZ^n)-H(Y_{1}^n|X_R^nZ^n)\nonumber\\
&\stackrel{(c)}{=}&H(Y_{1}^nY_{2}^n|Z^n)-H(Y_{2}^n|X_1^nX_2^nY_R^nZ^n)\nonumber\\
&&-H(Y_{1}^n|X_R^nX_1^nX_2^nY_R^nY_{2}^nZ^n)\nonumber\\
&\stackrel{(d)}{\geq}&H(Y_{1}^nY_{2}^n|Z^n)-H(Y_{1}^nY_{2}^n|X_1^nX_2^nY_R^nZ^n)\nonumber\\
&=&I(Y_{1}^nY_{2}^n;X_1^nX_2^nY_R^n|Z^n)\nonumber\\
&=&I(X_1^nX_2^n;Y_{1}^nY_{2}^n|Z^n)+I(Y_{1}^nY_{2}^n;Y_R^n|X_1^nX_2^nZ^n)\nonumber\\
&\stackrel{(e)}{\geq}&nR+I(Y_{1}^nY_{2}^n;Y_R^n|X_1^nX_2^nZ^n)-n\epsilon_n\nonumber\\
&\stackrel{(f)}{=}& nR+I(Y_{1}^n;Y_R^n|X_1^nZ^n)-n\epsilon_n\nonumber\\
&=&nR+\sum_{i=1}^nI(Y_{1}^n;Y_{Ri}|X_1^nY_{R1}^{i-1}Z^n)-n\epsilon_n\nonumber\\
&\geq& nR+\sum_{i=1}^nI(Y_{1i+1}^{n};Y_{Ri}|X_1^nY_{R1}^{i-1}Z^n)-n\epsilon_n\nonumber\\
&\stackrel{(g)}{\geq}&nR+\sum_{i=1}^nI(Y_{1i+1}^{n};Y_{Ri}|X_{1i}^nY_{R1}^{i-1}Z^n)-n\epsilon_n\nonumber\\
&=&nR+\sum_{i=1}^nI(\hat{Y}_{Ri};Y_{Ri}|X_{1i}U_iZ_i)-n\epsilon_n\nonumber,
\end{IEEEeqnarray}
where $(a)$ follows from (\ref{eq:R1Inequ}) and (\ref{eq:R0Inequ}); $(b)$ is due to the fact that conditioning
reduces entropy; $(c)$ is due to the Markov chains
$Y_{2}^n-(X_{2}^nZ^n)-(X_1^nY_R^n)$ and
$Y_{1}^n-(X_{R}^nZ^n)-(X_1^nX_2^nY_R^nY_{2}^n)$; $(d)$ follows since conditioning reduces entropy; $(e)$ is due to the expression  in
(\ref{eq:ncutset}); $(f)$ is due to the Markov chain
$(Y_R^nY_1^n)-(X_1^nZ^n)-(X_2^nY_2^n)$ and; $(g)$ is due to the
Markov chain $(Y_{1i+1}^{n}) - (X_{1i}^nY_{R1}^{i-1}Z^n)-
X_{11}^{i-1}$.

A single letter expression can be obtained by using the usual
time-sharing random variable arguments. Let $Q$ be a time sharing
random variable uniformly distributed over $\{1,...,n\}$,
independent of all the other random variables. Also, define a set of random variables
$(X_{1Q},Y_{RQ},U_Q,\hat{Y}_{RQ},Z_Q)$ satisfying
\begin{IEEEeqnarray}{rCl}
&\text{Pr}&\{X_{1Q}=x_1,Y_{RQ}=y_R,U_Q=u,\hat{Y}_{RQ}=\hat{y}_R,Z_Q=z|Q=i\}\nonumber\\
&=&\text{Pr}\{X_{1i}=x_1,Y_{Ri}=y_R,U_i=u,\hat{Y}_{Ri}=\hat{y}_D,Z_i=z\}\\
&& \qquad  \text{for } i=1,...,n.\nonumber
 \end{IEEEeqnarray}
  Define $U\triangleq(U_Q,Q)$, $\hat{Y}_R\triangleq\hat{Y}_{RQ}$, $X_1\triangleq X_{1Q}$,
 $Y_{RQ}\triangleq Y_R$ and $Z\triangleq Z_Q$. We note that the pmf of the tuple  $(X_{1},Y_{R},U,\hat{Y}_{R},Z)$ belongs to $\mathcal{P}$ in (\ref{eq:pmf}) as follows:
\begin{IEEEeqnarray}{rCl}
&p(u,& x_{1},y_{R},z,\hat{y}_{R})\nonumber\\
&=&
p(q,u_Q,x_{1Q},y_{RQ},z_{Q},\hat{y}_{RQ})\nonumber\\
&=&p(q,u_Q,x_{1Q})p(z_{Q}y_{RQ}\hat{y}_{RQ}|q,u_Qx_{1Q})\nonumber\\
&=&p(q,u_Q,x_{1Q})p(z_{Q}|q,u_Q,x_{1Q})p(y_{RQ},\hat{y}_{RQ}|q,u_Q,x_{1Q},z_{Q})\nonumber\\
&\stackrel{(a)}{=}&p(q,u_Q,x_{1Q})p(z)p(y_{RQ}|q,u_Q,x_{1Q},z_{Q})\nonumber\\
&&\cdot p(\hat{y}_{RQ}|q,u_Q,x_{1Q},z_{Q},y_{RQ})\nonumber\\
&\stackrel{(b)}{=}&p(q,u_Q,x_{1Q})p(z)p(y_{R}|x_{1},z)p(\hat{y}_{RQ}|q,u_Q,x_{1Q},z_{Q},y_{RQ})\nonumber\\
&\stackrel{(c)}{=}&p(q,u_Q,x_{1Q})p(z)p(y_{R}|x_{1},z)p(\hat{y}_{RQ}|q,u_Q,y_{RQ})\nonumber\\
&=&p(u, x_{1})p(z)p(y_{R}|x_1,z)p(\hat{y}_{R}|u,y_{R}),\nonumber
\end{IEEEeqnarray}
where $(a)$ follows since the channel state $Z^n$ is i.i.d; and thus $p(z_{Q}|q,u_Q,x_{1Q})=p(z_Q|q)=p(z)$; $(b)$ follows since $p(y_{RQ}|q,u_Q,x_{1Q},z_{Q})=p(y_{RQ}|q,x_{1Q},z_{Q})=p(y_R|x_1,z)$; $(c)$ follows from the Markov chain in (\ref{eq:Uchain}).

Then, we get the single letter expression,
\begin{IEEEeqnarray}{rCl}
R&\leq&R_2+\frac{1}{n}\sum_{i=1}^n[I(U_i;Y_{Ri})+I(X_{1i};\hat{Y}_{Ri}|U_iZ_i)]+\epsilon_n\nonumber\\
&=&R_2+I(U_Q;Y_{RQ}|Q)+I(X_{1Q};\hat{Y}_{RQ}|U_QZ_QQ)+\epsilon_n\nonumber\\
&\leq&R_2+I(U_QQ;Y_{RQ})+I(X_{1Q};\hat{Y}_{RQ}Q|U_QZ_Q)+\epsilon_n\nonumber\\
&=&R_2+I(U;Y_{R})+I(X_{1};\hat{Y}_{R}|UZ)+\epsilon_n\nonumber,
\end{IEEEeqnarray}
and
\begin{IEEEeqnarray}{rCl}
R_1+R_2&\geq&R+\frac{1}{n}\sum_{i=1}^nI(\hat{Y}_{Ri};Y_{Ri}|X_{1i}U_iZ_i)-n\epsilon_n\nonumber\\
&=&R+I(\hat{Y}_{RQ};Y_{RQ}|X_{1Q}U_QZ_QQ)-n\epsilon_n\nonumber\\
&=&R+I(\hat{Y}_{R};Y_{R}|X_{1}UZ)-n\epsilon_n\nonumber.
\end{IEEEeqnarray}

 The cardinality of the bounds on the alphabets of $U$ and $\hat{Y}_R$ can be found using the usual techniques \cite{elGamal:book}.
This completes the proof.

\section{Proof of Lemma \ref{lem:Conv2}}\label{app:Converse2}

Now, we will show that the expression of $R_{up}$ in (\ref{eq:Conv1}) is equivalent to the expression $\mathcal{C}$ in $(\ref{eq:capreg})$. First we will show that $\mathcal{C} \leq R_{up}$. Consider the
subset of pmf's in $\mathcal{P}$ such that
\begin{IEEEeqnarray}{rCl}
&R_1&+R_2-I(\hat{Y}_{R};Y_{R}|X_{1}UZ)\\
&&\geq R_2+I(U;Y_{R})+I(X_{1};\hat{Y}_{R}|UZ),
\end{IEEEeqnarray}
holds. Then, similarly to (\ref{eq:Eqcond}) in Appendix \ref{app:ProofLemmaAch} this condition is equivalent to
\begin{IEEEeqnarray}{rCl}
R_1\geq I(U;Y_{R})+I(Y_{R};\hat{Y}_{R}|UZ).
\end{IEEEeqnarray}
Hence, we have $\mathcal{C} \leq R_{up}$.

Then, it remains to show that $\mathcal{C} \geq R_{up}$. As $R_2$ can be
extracted from the supremum, it is enough to show that, for each  $(X_1,U,Z,Y_R,\hat{Y}_R)$  tuple with a joint pmf $p_e\in\mathcal{P}$ satisfying
\begin{IEEEeqnarray}{rCl}\label{eq:CondA}
&R(p_e)&\leq I(U;Y_R)+I(X_1;\hat{Y}_R|UZ),\nonumber\\
&\text{where }& R(p_e)\triangleq R_1-I(\hat{Y}_R;Y_R|X_1UZ),
\end{IEEEeqnarray}
there exist random variables
$(X_1^*,U^*,Z,Y_R^*,\hat{Y}_R^*)$ with joint pmf $p_e^*\in\mathcal{P}$
that satisfy
\begin{IEEEeqnarray}{rCl}\label{eq:newcondAle2}
R(p_e)&=&I(U^*;Y_R)+I(X_1^*;\hat{Y}_R^*|U^*Z) \;\text{and }\nonumber\\
R(p_e)&\leq& R_1-I(\hat{Y}_R^*;Y_R|X_1^*U^*Z).
\end{IEEEeqnarray}
This argument is proven next.

 Let $B$ denote a Bernoulli random
variable with parameter $\lambda\in[0,1]$, i.e., $B=1$ with
probability $\lambda$, and $B=0$ with probability $1-\lambda$. We
define the triplets of random variables:
\begin{IEEEeqnarray}{rCl}
(U^{'},X_1^{'},\hat{Y}_R^{'})=\begin{cases}
(U,X_1,\hat{Y}_R)&\text{if }B=1,\\
(X_1,X_1,\emptyset)&\text{if }B=0,\\
\end{cases}
\end{IEEEeqnarray}
and
\begin{IEEEeqnarray}{rCl}
(U^{''},X_1^{''},\hat{Y}_R^{''})=\begin{cases}
(X_1,X_1,\emptyset)&\text{if }B=1,\\
(\emptyset,\emptyset,\emptyset)&\text{if }B=0.\\
\end{cases}
\end{IEEEeqnarray}
We first consider the case $R(p_e)> I(X_1;Y_R)$. Let
$U^*=(U^{'},B)$, $X_1^*= X_1^{'}$, $\hat{Y}_R^*=(\hat{Y}_R^{'},B)$.
For $\lambda=1$,
\begin{IEEEeqnarray}{lCl}
I(U^*;Y_R)+I(X_1^*;\hat{Y}_R^*|U^*Z)\nonumber
&=&I(U;Y_R)+I(X_1;\hat{Y}_R|UZ)\\
&>&R(p_e),
\end{IEEEeqnarray}
 and for $\lambda=0$,
 \begin{IEEEeqnarray}{rCl}
I(U^*;Y_R)+I(X_1^*;\hat{Y}_R^*|U^*Z)&=&I(X_1;Y_R)\nonumber\\
&<&R(p_e).
\end{IEEEeqnarray}
As $I(U^*;Y_R)+I(X_1^*;\hat{Y}_R^*|U^*Z)$ is a continuous function
of $\lambda$, by the intermediate value theorem, there exists a
$\lambda\in[0,1]$ such that
$I(U^*;Y_R)+I(X_1^*;\hat{Y}_R^*|U^*Z)=R(p_e)$. We denote the corresponding joint distribution by $p_e^*$.

We have
\begin{IEEEeqnarray}{rCl}
I(\hat{Y}_R^*;Y_R|X_1^*U^*Z)&=&I(\hat{Y}_R^{'};Y_R|X_1^{'}U^{'}Z B)\nonumber\\
&=&\lambda I(\hat{Y}_R;Y_R|X_1UZ)\nonumber\\
&\leq&I(\hat{Y}_R;Y_R|X_1 U Z),
\end{IEEEeqnarray}
which implies that $p_e^*$ satisfies (\ref{eq:newcondAle2}) since
\begin{IEEEeqnarray}{rCl}
R(p_e)&=&R_1-I(\hat{Y}_R;Y_R|X_1UZ)\nonumber\\
&\leq&R_1-I(\hat{Y}_R^*;Y_R|X_1^*U^*Z).
\end{IEEEeqnarray}
Next we consider the case $R(p_e)\leq I(X_R;Y_1)$. We define
$U^*=(U^{''},B)$, $X_1^*= X_1^{''}$ and
$\hat{Y}_R^*=(\hat{Y}_R^{''},B)$. Then, for $\lambda=1$,
\begin{IEEEeqnarray}{rCl}
I(U^*;Y_R)+I(X_1^*;\hat{Y}_R^*|U^*Z)&=&I(X_1;Y_R)\nonumber\\
&\geq& R(p_e),\nonumber
\end{IEEEeqnarray}
 and for $\lambda=0$,
 \begin{IEEEeqnarray}{rCl}
I(U^*;Y_R)+I(X_1^*;\hat{Y}_R^*|U^*Z)&=&0\nonumber\\
&<&R(p_e).
\end{IEEEeqnarray}

Once again, as $I(U^*;Y_R)+I(X_1^*;\hat{Y}_R^*|U^*Z)$ is a
continuous function of $\lambda$, by the intermediate value theorem,
there exists a $\lambda\in[0,1]$ such that
$I(U^*;Y_R)+I(X_1^*;\hat{Y}_R^*|U^*Z)=R(p_e)$. Again, we denote this joint distribution by $p_e^*$. On the other hand,
we have $I(\hat{Y}_R^*;Y_R|X_1^*U^*Z)=0$, which implies that
\begin{IEEEeqnarray}{rCl}
R(p_e)&=&R_1-I(\hat{Y}_R;Y_R|X_1UZ)\nonumber\\
&\leq&R_1\nonumber\\
&=&R_1-I(\hat{Y}_R^*;Y_R|X_1^*U^*Z).
\end{IEEEeqnarray}
That is, $p_e^*$ also satisfies (\ref{eq:newcondAle2}).

We have shown that for any joint pmf $p_e\in\mathcal{P}$
satisfying $(\ref{eq:CondA})$, there exist another joint pmf, $p_e^*$, that satisfies
$(\ref{eq:newcondAle2})$. For a distribution satisfying $(\ref{eq:newcondAle2})$ we
can write
\begin{IEEEeqnarray}{rCl}
R_1&>&I(U^*;Y_R)+I(X_1^*;\hat{Y}_R^*|U^*Z)+I(\hat{Y}_R^*;Y_R|X_1^*U^*Z)\nonumber\\
&=&I(U^*;Y_R)+I(Y_RX_1^*;\hat{Y}_R^*|U^*Z)\nonumber\\
&\stackrel{(a)}{=}&I(U^*;Y_R)+I(\hat{Y}_R^*;Y_R|U^*Z)\nonumber
\end{IEEEeqnarray}
where $(a)$ is due to Markov chain $X_1^*-(Y_RZU^*)-\hat{Y}_R^*$.
 This concludes the proof.

\section{Proof of Lemma \ref{lemm:CF}}\label{app:CF}

Before deriving the maximum achievable rate by CF in Lemma \ref{lemm:CF}, we provide some definitions that will be used in the proof.

Let $X$ and $Y$ be a pair of discrete random variables, where $\mathcal{X}= \{1,2,...,n\}$ and $\mathcal{Y}= \{1,2,...,m\}$, for $n,m<\infty$. Let $\mathbf{p}_Y\in \Delta_{m}$ denote the distribution of  $Y$, where $\Delta_k$ denotes the $(k-1)$-dimensional simplex of probability $k$-vectors. We define $T_{XY}$ as the $n\times m$ stochastic matrix with entries $T_{XY}(j,i)=\text{Pr}\{X=j|Y=i\}$. Note that the joint distribution $p(x,y)$ is characterized by $T_{XY}$ and $\mathbf{p}_Y$.

Next, we define the conditional entropy bound from \cite{Witsenhausen:IT:75}, which lower bounds the conditional entropy between two variables. Note the relabeling of the variables in \cite{Witsenhausen:IT:75} to fit our model.

\begin{definition}[Conditional Entropy Bound]
Let $\mathbf{p}_Y\in\Delta_{m}$ be the distribution of $Y$ and $T_{XY}$ denote the channel matrix relating $X$ and $Y$.  Then, for $\mathbf{q}\in\Delta_{m}$ and $0\leq s\leq H(Y)$, define the function
\begin{IEEEeqnarray}{rCl}\label{eq:ConditionalEntropyBound}
F_{T_{XY}}(\mathbf{q},s)&\triangleq& \inf_{\substack{p(w|y):\;X-Y-W, \\ H(Y|W)= s, \; \mathbf{p}_Y=\mathbf{q}.}} H(X|W).
\end{IEEEeqnarray}
\end{definition}
That is, $F_{T_{XY}}(\mathbf{q},s)$ is the infimum of $H(X|W)$ given a specified  distribution $\mathbf{q}$ and the value of $H(Y|W)$. Many properties of $F_{T_{XY}}(\mathbf{q},s)$ are derived in \cite{Witsenhausen:IT:75}, such as its convexity on $(\mathbf{q},s)$ \cite[Theorem 2.3]{Witsenhausen:IT:75} and its non-decreasing monotonicity in $s$ \cite[Theorem 2.5]{Witsenhausen:IT:75}.

Consider a sequence of $N$ random variables $\mathbf{Y}=(Y_1,...,Y_N)$ and
denote by $\mathbf{q}_i$ the distribution of $Y_i$, for $i=1,...,N$, by $\mathbf{q}^{(N)}$ the joint distribution of $\mathbf{Y}$ and by $\mathbf{q}\triangleq\frac{1}{N}\sum_{i=1}^{N}\mathbf{q}_i$ the average distribution.   Note that $Y_1,...,Y_N$ can have arbitrary correlation.
 Define the sequence $\mathbf{X}=(X_1,...,X_N)$, in which  $X_i$, $i=1,...,N$, is jointly distributed with each $Y_i$ through the stochastic matrix $T_{XY}$ and denote by $T_{XY}^{(N)}$ the Kronecker product of $N$ copies of the stochastic matrix $T_{XY}$.

 Then, the theorem given in \cite[Theorem 2.4]{Witsenhausen:IT:75} can be straightforwardly generalized to non i.i.d. sequences as given in the following lemma.
\begin{lemma}\label{lem:FIneqN}
For $N\in\mathds{Z}^+$, and $0\leq Ns\leq H(\mathbf{Y})$, we have
\begin{IEEEeqnarray}{rCl}
F_{T_{XY}^{(N)}}(\mathbf{q}^{(N)},Ns)\geq NF_{T_{XY}}(\mathbf{q},s),
\end{IEEEeqnarray}
where equality holds for i.i.d. $Y_i$ components following $\mathbf{q}$.
\end{lemma}

\begin{proof}

Let $W,\mathbf{X},\mathbf{Y}$ be a Markov chain, such that $H(\mathbf{Y}|W)=Ns$. Then, using the standard identity we have
\begin{IEEEeqnarray}{rCl}
H(\mathbf{Y}|W)=\sum_{k=1}^NH(Y_k|\mathbf{Y}_1^{k-1},W),\\
H(\mathbf{X}|W)=\sum_{k=1}^NH(X_k|\mathbf{X}_1^{k-1},W).\label{eq:Inqcond1}
\end{IEEEeqnarray}
Letting $s_k=H(Y_k|\mathbf{Y}_1^{k-1},W)$, we have
\begin{IEEEeqnarray}{rCl}\label{eq:AverageS}
\frac{1}{N}\sum_{k=1}^{N}s_k=s.
\end{IEEEeqnarray}

Also, from the Markov chain $X_k-(\mathbf{Y}_1^{k-1},W)-\mathbf{X}_1^{k-1}$, we have
\begin{IEEEeqnarray}{rCl}\label{eq:Inqcond2}
H(X_k|\mathbf{X}_1^{k-1},W)&\geq& H(X_k|\mathbf{Y}_1^{k-1},\mathbf{X}_1^{k-1},W)\\
&=&H(X_k|\mathbf{Y}_1^{k-1},W).
\end{IEEEeqnarray}
Applying the conditional entropy bound in (\ref{eq:ConditionalEntropyBound}) we have
\begin{IEEEeqnarray}{rCl}\label{eq:Inqcond3}
H(X_k|\mathbf{Y}_1^{k-1},W)\geq F_{T_{XY}}(\mathbf{q}_k,s_k).
\end{IEEEeqnarray}
Combining (\ref{eq:Inqcond1}), (\ref{eq:Inqcond2}) and (\ref{eq:Inqcond3}) we have
\begin{IEEEeqnarray}{rCl}
H(\mathbf{X}|W)&\geq& \sum_{k=1}^NF_{T_{XY}}(\mathbf{q}_k,s_k)\nonumber\geq NF_{T_{XY}}(\mathbf{q},s),
\end{IEEEeqnarray}
where the last inequality follows from the convexity of $F_{T}(\mathbf{q},s)$ in $\mathbf{q}$ and $s$ and (\ref{eq:AverageS}).

If we let $W,\mathbf{X}$, $\mathbf{Y}$ be $N$ independent copies of the random variables $W,X,Y$, that achieve $F_{T_{XY}}(\mathbf{q},s)$, we have $H(\mathbf{Y}|W)=Ns$ and $H(\mathbf{X}|W)=F_{T_{XY}^{(N)}}(\mathbf{q}^N)=NF_{T_{XY}}(\mathbf{q},s)$. Therefore,  $F_{T_{XY}^{(N)}}(\mathbf{q}^N)\leq NF_{T_{XY}}(\mathbf{q},s)$, and the equality holds for i.i.d. components of $\mathbf{Y}$.
\end{proof}

Now, we look into the binary symmetric channel $Y=X\oplus N$ where $N\sim \text{Ber}(\delta)$. Due to the binary modulo-sum operation, we have $X=Y\oplus N$, and we can characterize the channel $T_{XY}$ of this model as
\begin{IEEEeqnarray}{rCl}\label{eq:TxyBinMat}
T_{XY}=\left[
\begin{matrix}
1-\delta&\delta\\ \delta & 1-\delta
\end{matrix}
\right].
\end{IEEEeqnarray}

When $Y$ and $X$ are related through channel $T_{XY}$ in (\ref{eq:TxyBinMat}), $F_{T_{XY}}(\mathbf{q},s)$ is characterized as follows \cite{Witsenhausen:IT:75}.
\begin{lemma}\label{lem:binaryF}
Let $Y\sim \text{Ber}(q)$, i.e., $\mathbf{q}=[q,1-q]$, and $T_{XY}$ be given as in (\ref{eq:TxyBinMat}). Then the conditional entropy bound is
\begin{IEEEeqnarray}{rCl}
F_{T_{XY}}(q,s)=h_2(\delta\star h^{-1}_2(s)), \quad \text{for } 0\leq s \leq h_2(q).\nonumber
\end{IEEEeqnarray}

\end{lemma}

In the following, we use the properties of $F_{T_{XY}}(\mathbf{q},s)$ to derive the maximum rate achievable by CF in the  parallel binary symmetric MRC-D. From (\ref{eq:CFparallelRate}), we have
\begin{IEEEeqnarray}{rCl}
I(Y_R^1,Y_R^2;\hat{Y}_R|Z)&=&I(X_1^1\oplus N_1\oplus Z, X_1^2\oplus N_2;\hat{Y}_R|Z)\nonumber\\
&=&I(X_1^1\oplus N_1, X_1^2\oplus N_2;\hat{Y}_R|Z).\nonumber
\end{IEEEeqnarray}

Let us define $\bar{Y}_R^1\triangleq X_1^1\oplus N_1$ and $\bar{\mathbf{Y}}_R \triangleq (\bar{Y}_R^1,Y_R^2)$, and the channel input $\mathbf{X}\triangleq (X_1^1,X_1^2)$. Note that the distribution of $\bar{\mathbf{Y}}_R$, given by $\mathbf{q}^{(2)}$, determines the distribution of $\mathbf{X}$ via $T_{XY}^{(2)}$, the Kronecker product of $T_{XY}$ in (\ref{eq:TxyBinMat}). Then, we can rewrite the achievable rate for CF in (\ref{eq:CFparallelRate}) as follows
\begin{IEEEeqnarray}{rCl}\label{eq:CFrateProof}
R_{\mathrm{CF}} &=& \max_{{p(\mathbf{x})p(z)p(\bar{\mathbf{y}}_R|\mathbf{x})p(\hat{y}_R|\bar{\mathbf{y}}_R,z)}} I(\mathbf{X}, \hat{Y}_R|Z)\nonumber\\
&&\text{ s.t. } R_1 \geq I(\bar{\mathbf{Y}}_R; \hat{Y}_R|Z).
\end{IEEEeqnarray}

Next, we derive a closed form expression for $R_{\mathrm{CF}}$. First, we note that if $R_1\geq2$, we have $H(\bar{\mathbf{Y}}_R)\leq R_1$ and  $R_{\mathrm{CF}}=2(1-h(\delta))$, i.e., CF meets the cut-set bound.

For fixed $\mathbf{q}^{(2)}$, if $H(\bar{\mathbf{Y}}_R )\!\leq R_1\!\leq 2$, the constraint in (\ref{eq:CFrateProof}) is satisfied by any $\hat{Y}_R$, and can be ignored. Then, due to the Markov chain $\mathbf{X}-\bar{\mathbf{Y}}_R -\hat{Y}_RZ$, and the data processing inequality, the achievable rate is upper bounded by
\begin{IEEEeqnarray}{rCl}\label{eq:Bound1}
R_{\mathrm{CF}} &\leq& I(\mathbf{X}, \bar{\mathbf{Y}}_R )=H(\bar{\mathbf{Y}}_R )-2h(\delta)\leq R_1-2h(\delta).
\end{IEEEeqnarray}

\interdisplaylinepenalty=1000

For $R_1\leq H(\bar{\mathbf{Y}}_R)\leq 2$, the achievable rate by CF is upper bounded as follows.
\begin{IEEEeqnarray}{rCl}
R_{\mathrm{CF}} &\stackrel{(a)}{=}& \max_{{p(\mathbf{x})p(z)p(\bar{\mathbf{y}}_R|\mathbf{x})p(\hat{y}_R|\bar{\mathbf{y}}_R,z)}}  H(\mathbf{X})-H(\mathbf{X}| Z\hat{Y}_R)\nonumber\\
&&\text{s.t. } H(\bar{\mathbf{Y}}_R|Z \hat{Y}_R) \geq H(\bar{\mathbf{Y}}_R)-R_1\nonumber\\
&\stackrel{(b)}{\leq}& \max_{{p(\mathbf{x})p(\bar{\mathbf{y}}_R|\mathbf{x})p(w|\bar{\mathbf{y}}_R)}} H(\mathbf{X})-H(\mathbf{X}| W)\nonumber\\
&&\text{s.t. } H(\bar{\mathbf{Y}}_R|W) \geq H(\bar{\mathbf{Y}}_R)-R_1\nonumber\nonumber\\
&=& \max_{{p(\mathbf{x})p(\bar{\mathbf{y}}_R|\mathbf{x})}} [H(\mathbf{X})-\min_{{p(w|\bar{\mathbf{y}}_R)}}H(\mathbf{X}| W)]\nonumber\\
&&\text{s.t. } H(\bar{\mathbf{Y}}_R|W) \geq H(\bar{\mathbf{Y}}_R)-R_1\nonumber\\
&\stackrel{(c)}{=}& \max_{{p(\mathbf{x})p(\bar{\mathbf{y}}_R|\mathbf{x}),0\leq s\leq H(\bar{\mathbf{Y}}_R)}} [H(\mathbf{X})-F_{T_{XY}^{(2)}}(\mathbf{q}^{(2)},s)]\nonumber\\
&&\text{ s.t. } s\geq H(\bar{\mathbf{Y}}_R)-R_1\nonumber\\
&\stackrel{(d)}{=}& \max_{{p(\mathbf{x})p(\bar{\mathbf{y}}_R|\mathbf{x})}} [H(\mathbf{X})-F_{T_{XY}^{(2)}}(\mathbf{q}^{(2)},H(\bar{\mathbf{Y}}_R)-R_1)]\nonumber\\
&\stackrel{(e)}{\leq}& \max_{{p(\mathbf{x})p(\bar{\mathbf{y}}_R|\mathbf{x})}}  [H(\mathbf{X})-2F_{T_{XY}}(\mathbf{q}, (H(\bar{\mathbf{Y}}_R)-R_1)/2)]\nonumber,
\end{IEEEeqnarray}
where $(a)$ follows from the independence of $Z$ from $\mathbf{X}$ and $\bar{\mathbf{Y}}_R$; $(b)$ follows since optimizing over $W$  can only increase the value compared to optimizing over $(Z,\hat{Y}_R)$; $(c)$ follows from the definition of the conditional entropy bound in (\ref{eq:ConditionalEntropyBound}); $(d)$ follows from the nondecreasing monotonicity of $F_{T_{XY}^{(2)}}(\mathbf{q}^{(2)},s)$ in $s$; and $(e)$ follows from Lemma \ref{lem:FIneqN}, and $\mathbf{q}\triangleq[q,1-q]=\frac{1}{2}(\mathbf{q}_1+\mathbf{q}_2)$ is the average distribution of $\mathbf{Y}$. 

Now, we lower bound $H(\bar{\mathbf{Y}}_R)$. Since conditioning reduces entropy, we have $H(\bar{\mathbf{Y}}_R)\geq H(\bar{\mathbf{Y}}_R|N_1N_2)=H(\mathbf{X})$, and then we can lower bound $H(\bar{\mathbf{Y}}_R)$ as follows:
\begin{IEEEeqnarray}{rCl}\label{eq:HYInequality}
\max\{H(\mathbf{X}),R_1\}\leq H(\bar{\mathbf{Y}}_R)\leq  2.
\end{IEEEeqnarray}

Let $\nu \triangleq (\max\{H(\mathbf{X}),R_1\}-R_1)/2$. Then, we have
\begin{IEEEeqnarray}{rCl}
&R_{\mathrm{CF}}&\\
&\stackrel{(a)}{\leq}&
\max_{{p(\mathbf{x})p(\bar{\mathbf{y}}_R|\mathbf{x})}}  [H(\mathbf{X})-2F_{T_{XY}}(\mathbf{q}, \nu)]\nonumber,\\
&\stackrel{(b)}{=}& \max_{{p(\mathbf{x})p(\bar{\mathbf{y}}_R|\mathbf{x})}} [H(\mathbf{X})-2h_2(\delta\star h^{-1}_2(\nu))]\nonumber\\
&&\text{s.t. } 0\leq \nu \leq h_2(q) \nonumber\\
&\stackrel{(c)}{\leq}& \max_{{p(\mathbf{x})}} [H(\mathbf{X})-2h_2(\delta\star h^{-1}_2(\nu))]\nonumber\\
&&\text{s.t. } R_1\leq \max\{H(\mathbf{X}),R_1\} \leq2+R_1 \nonumber\\
&\stackrel{(d)}{=}& \max_{{p(\mathbf{x})}} [H(\mathbf{X})-2h_2(\delta\star h^{-1}_2(\max\{H(\mathbf{X}),R_1\}-R_1)/2))]\nonumber\\
&&\text{s.t. } \max\{H(\mathbf{X}),R_1\} \leq2 \nonumber\\
&\stackrel{(e)}{=}& \max_{0\leq\alpha\leq1} [2\alpha-2h_2(\delta\star h^{-1}_2((\max\{2\alpha,R_1\}-R_1)/2))]\nonumber\\
&&\text{s.t. } \max\{R_1,2\alpha\}\leq 2 \nonumber,
\end{IEEEeqnarray}
where  $(a)$ follows from (\ref{eq:HYInequality}) and $F_{T_{XY}}(\mathbf{q}, s)$ being non-decreasing in $s$; equality $(b)$ follows from the definition of $F_{T_{XY}}(q,s)$ for the binary symmetric channel;  $(c)$ follows since $h_2(q)\leq 1$, and we are enlarging the optimization domain; $(d)$ follows since there is no loss in generality by reducing the optimization set, since $\max\{H(\mathbf{X}),R_1\}\geq R_1$  and from (\ref{eq:HYInequality}), any $(\mathbf{X},\bar{\mathbf{Y}}_R)$ following $p(\mathbf{x},\bar{\mathbf{y}}_R)$ satisfy $\max\{H(\mathbf{X}),R_1\}\leq 2$; and $(e)$ follows from defining $H(\mathbf{X})\triangleq2\alpha$, for $0\leq \alpha\leq 1$.

Then, for $2\alpha\leq R_1$, we have
\begin{IEEEeqnarray}{rCl}\label{eq:RFBound1}
R_{\mathrm{CF}}&\leq&  \max_{0\leq\alpha\leq R_1/2}[ 2\alpha-2h_2(\delta)]=R_1-2h_2(\delta),
\end{IEEEeqnarray}
and for $2\alpha> R_1$, we have
\begin{IEEEeqnarray}{rCl}\label{eq:RFBound3}
R_{\mathrm{CF}}&\leq& \max_{R_1/2<\alpha\leq1}[ 2\alpha-2h_2(\delta\star h^{-1}_2(\alpha-R_1/2))].
\end{IEEEeqnarray}

Now, we solve (\ref{eq:RFBound3}). Let us define $f(u)\triangleq h_2(\delta\star h^{-1}_2(u))$ for $0\leq u\leq1$. Then, we have the following lemma from  \cite{Wyner1973AtheoremEntropyBinary}.
\begin{lemma}[\cite{Wyner1973AtheoremEntropyBinary}]\label{lemm:fuConvex}
Function $f(u)$ is convex for $0\leq u\leq1$.
 \end{lemma}

 We define $g(\alpha)\triangleq \alpha-h_2(\delta\star h^{-1}_2(\alpha-R_1/2))$, such that $R_{\mathrm{CF}}\leq \max_{R_1/2<\alpha\leq1}2g(\alpha)$. We have that $g(\alpha)$ is concave in $\alpha$, since is a shifted version by $\alpha$, which is linear, of the composition of the concave function $-f(u)$ and the affine function $\alpha-R_1/2$.

\begin{proposition}\label{prop:galphaIncreasing}
$g(\alpha)$ is monotonically increasing for $R_1/2\leq \alpha\leq 1+R_1/2$.
\end{proposition}
\begin{proof}
Using the chain rule for composite functions, we have
 \begin{IEEEeqnarray}{rCl}
 \frac{d^2g(\alpha)}{d\alpha^2}=-f''(\alpha-R_1/2),
 \end{IEEEeqnarray}
where $f''(u)\triangleq d^2f/du^2(u)$.

Since $g(\alpha)$ is convex, and is defined over a convex region, it follows that its unique maximum is achieved either for $f''(\alpha-R_1/2)=0$, or at the boundaries of the region. It is shown in \cite[Lemma 2]{Wyner1973AtheoremEntropyBinary} that $f''(u)>0$ for $0< u<1$. Therefore, the maximum is achieved either at $u=0$ or at $u=1$, or equivalently, for $\alpha=R_1/2$ or $\alpha=1+R_1/2$. Since $g(R_1/2)=R_1/2-h_2(\delta)$ and $g(1+R_1/2)=R_1/2$, i.e., $g(R_1/2)<g(1+R_1/2)$, it follows that $g(\alpha)$ is monotonically increasing in $\alpha$ for $R_1/2\leq \alpha\leq 1+R_1/2$.
\end{proof}

From Proposition \ref{prop:galphaIncreasing} if follows that for $R_1/2\leq\alpha\leq 1$, $g(\alpha)$ achieves its maximum at $\alpha=1$. Then, for $2\alpha> R_1$, we have
\begin{IEEEeqnarray}{rCl}\label{eq:RFBound2}
R_{\mathrm{CF}}&\leq&  2(1-h_2(\delta\star h^{-1}_2(1-R_1/2))).
\end{IEEEeqnarray}

Thus, from (\ref{eq:RFBound1}) and (\ref{eq:RFBound2}),  for $R_1\leq H(\bar{\mathbf{Y}}_R)$ we have 
\begin{IEEEeqnarray}{rCl}\label{CFparallBound}
R_{\mathrm{CF}}&\leq& 2\max\{ R_1/2-h_2(\delta),1-h_2(\delta\star h^{-1}_2(1-R_1/2))\}\nonumber\\
&=&2(1-h_2(\delta\star h^{-1}_2(1-R_1/2))),
\end{IEEEeqnarray}
where the equality follows from Proposition \ref{prop:galphaIncreasing} by noting that the first element in the maximum coincides with $g(R_1/2)=R_1/2-h_2(\delta)$, and the second one coincides with $g(1)$.

Finally, $R_{\mathrm{CF}}$ is upper bounded by the maximum over the joint distributions satisfying $H(\bar{\mathbf{Y}}_R)\leq R_1$  given in (\ref{eq:Bound1}), and the upper bound for the joint distributions satisfying $R_1\leq H(\bar{\mathbf{Y}}_R)$ given in (\ref{CFparallBound}). Since (\ref{eq:Bound1}) coincides with $g(R_1/2)$, (\ref{CFparallBound}) serves as an upper bound on $R_{\mathrm{CF}}$ when $R_1\leq H(\bar{\mathbf{Y}}_R)$.

Next, we show that the upper bound in (\ref{CFparallBound}) is achievable by considering the following variables
\begin{IEEEeqnarray}{rCl}
X_1^1&\sim&\text{Ber}(1/2),\quad X_1^2\sim\text{Ber}(1/2),\quad \hat{Y}_R=(\hat{Y}_R^1,\hat{Y}_R^2)\nonumber\\
\hat{Y}_R^1&=&Y_R^1\oplus Q_1,\; Q_1\sim\text{Ber}(h^{-1}_2(1-R_1/2)).\nonumber\\
\hat{Y}_R^2&=&Y_R^2\oplus Q_2,\; Q_2\sim\text{Ber}(h^{-1}_2(1-R_1/2)).\nonumber
\end{IEEEeqnarray}
Let $Q_i\sim\text{Ber}(\nu)$ for $i=1,2$. Then from the constraint in (\ref{eq:CFparallelRate}) we have
\begin{IEEEeqnarray}{rCl}
&I(Y_R^1&,Y_R^2;\hat{Y}_R|Z)\nonumber\\
&=&H(\hat{Y}_R|Z)-H(\hat{Y}_R|Y_R^1Y_R^2Z)\nonumber\\
&=&H(X_1^1\oplus N_1\oplus Q_1,X_1^2\oplus N_2\oplus Q_2)-H(Q_1,Q_2)\nonumber\\
&\stackrel{(a)}{=}&2-2h_2(\nu),\nonumber
\end{IEEEeqnarray}
where $(a)$ follows since $X_1^i\sim\text{Ber}(1/2)$, $i=1,2$ and from the independence of $Q_1$ and $Q_2$. We have $2h_2(\nu)\geq 2-R_1$, and thus, $\nu\geq h^{-1}_2(1-R_1/2)$.

Then, the achievable rate in (\ref{eq:CFparallelRate})  is given by
\begin{IEEEeqnarray}{rCl}
I(\mathbf{X};\hat{Y}_R|Z)&=&H(\hat{Y}_R|Z)-H(\hat{Y}_R|\mathbf{X}Z)\nonumber\\
&=&H(X_1^1\oplus N_1\oplus Q_1,X_1^2\oplus N_2\oplus Q_2)\nonumber\\
&&-H(N_1\oplus Q_1,N_2\oplus Q_2)\nonumber\\
&=&2-2h(\delta\star \nu)\nonumber\\
&\leq &2-2h_2(\delta\star h^{-1}_2(1-R_1/2)),\nonumber
\end{IEEEeqnarray}
where the last inequality follows from the bound on $\nu$. This completes the proof.

\section{Proof of Lemma \ref{lemm:pDCFRate}}\label{app:pDCFRate}

From (\ref{eq:capreg}), the achievable rate for the proposed pDCF scheme is given by
\begin{IEEEeqnarray}{rCl}
R_{\mathrm{pDCF}}&=&I(X_1^1;Y_R^1)+I(X_1^2;\hat{Y}_R|Z)\nonumber\\
&\text{s.t. }&R_1\geq I(X_1^1;Y_R^1)+I(Y_R^2;\hat{Y}_R|Z).\nonumber
\end{IEEEeqnarray}

First, we note that the constraint is always satisfied for the choice of variables:
\begin{IEEEeqnarray}{rCl}
&I(X_1^1&;Y_R^1)+I(Y_R^2;\hat{Y}_R|Z)\nonumber\\
&=&H(Y_R^1)-H(N_1)+H(X_1^2\oplus N_2\oplus Q)-H(Q)\nonumber\\
&=&1-h_2(\delta)+1-h_2(h^{-1}_2(2-h(\delta)-R_1))\nonumber\\
&=&R_1,
\end{IEEEeqnarray}
where $H(Y_R^1)=1$ since $X_1^1\sim\text{Ber}(1/2)$ and $H(X_1^2\oplus N_2\oplus Q)=1$ since $X_1^2\sim\text{Ber}(1/2)$.
Then, similarly the achievable rate is given by
\begin{IEEEeqnarray}{rCl}
R_{\mathrm{pDCF}}&=&I(X_1^1;Y_R^1)+I(X_1^2;\hat{Y}_R|Z)\nonumber\\
&=&H(Y_R^1)-H(N_1)+H(X_1^2\oplus N_2\oplus Q)-H(V\oplus Q)\nonumber\\
&=&1-h_2(\delta)+1-h_2(\delta\star h^{-1}_2(2-h(\delta)-R_1))\nonumber,
\end{IEEEeqnarray}
which completes the proof.

\section{Proof of Lemma \ref{lemm:pDCFGauss}}\label{app:pDCFGauss}
 By evaluating (\ref{eq:capreg}) with the considered Gaussian random variables, we get
\begin{IEEEeqnarray}{rCl}\label{eq:GaussR}
R&=&\frac{1}{2}\log\left(1+\frac{\alpha P}{\bar{\alpha}P+1}\right)\left(1+\frac{\bar{\alpha}P}{(1-\rho^2)+\sigma^2_q}\right) \nonumber\\
&&\text{s.t. }R_1\geq\frac{1}{2}\log\left(1+\frac{\alpha
P}{\bar{\alpha}P+1}\right)\left(1+\frac{\bar{\alpha}P+(1-\rho^2)}{\sigma^2_q}\right)\nonumber.
\end{IEEEeqnarray}
We can rewrite the constraint on $R_1$ as,
\begin{IEEEeqnarray}{rCl}
\sigma^2_q\geq
f(\alpha)\triangleq\frac{(P+1)(\bar{\alpha}P+1-\rho^2)}{2^{2R_1}(\bar{\alpha}P+1)-(P+1)}.
\end{IEEEeqnarray}
Since $R$ is increasing in $\sigma^2_q$, it is clear that the
optimal $\sigma_q^2$ is obtained by $\sigma^2_q=f(\alpha)$, where
$\alpha$ is chosen such that $f(\alpha)\geq 0$. It is easy to check
that $f(\alpha)\geq0$ for
\begin{IEEEeqnarray}{rCl}
\alpha\in\left[0,\min\left\{(1-2^{-2R_1})\left(1+\frac{1}{P}\right),1\right\}\right].
\end{IEEEeqnarray}
Now, we substitute $\sigma^2_q=f(\alpha)$ in (\ref{eq:GaussR}), and
write the achievable rate as a function of $\alpha$ as
\begin{IEEEeqnarray}{rCl} R(\alpha)=\frac{1}{2}\log G(\alpha),
\end{IEEEeqnarray}
where
\begin{IEEEeqnarray}{rCl}
G(\alpha)&\triangleq&\left(1+\frac{\alpha
P}{\bar{\alpha}P+1}\right)\left(1+\frac{\bar{\alpha}P}{(1-\rho^2)+f(\alpha)}\right)\nonumber\\
&=&\frac{2^{2R_1}(1+P)(1-\rho^2+\bar{\alpha}P)}{(1-\rho^2)2^{2R_1}(1+\bar{\alpha}P)+\bar{\alpha}P(1+P)}.
\end{IEEEeqnarray}
We take the derivative of $G(\alpha)$ with respect to $\alpha$:
\begin{IEEEeqnarray}{rCl}
G'(\alpha)\triangleq\frac{2^{2R_1} P (1+P) \left(1-\rho ^2\right)
\left(P+1-2^{2R_1} \rho ^2\right)}{\left[P
(1+P)\bar{\alpha}+2^{2R_1} (1+ \bar{\alpha}P) \left(1-\rho
^2\right)\right]^2}\nonumber.
 \end{IEEEeqnarray}
We note that if $\rho^2\geq2^{-R_1}(P+1)$, then $G'(\alpha)<0$, and
hence, $G(\alpha)$ is monotonically decreasing. Achievable rate $R$
is maximized by setting $\alpha^*=0$. When $\rho^2<2^{-R_1}(P+1)$,
we have $G'(\alpha)>0$, and hence
$\alpha^*=\min\left\{(1-2^{-R_1})\left(1+\frac{1}{P}\right),1\right\}=1$, since we have
$(1-2^{-R_1})\left(1+\frac{1}{P}\right)\geq(1+\frac{1-\rho^2}{P})>1$.

\section{Proof of Lemma \ref{lemm:CutCap}}\label{app:CutCap}

In order to characterize the capacity of the binary symmetric MRC-D, we find the optimal distribution of $(U,X_1,\hat{Y}_R)$ in Theorem \ref{th:capreg} for $Z\sim\text{Ber}(1/2)$.
First, we note that $U$ is independent of $Y_R$ since
\begin{IEEEeqnarray}{rCl}
I(U;Y_R)\leq I(X_1;Y_R)=0,
\end{IEEEeqnarray}
where  the inequality follows from the Markov chain $U-X_1-Y_R$, and the equality follows since for $Z\sim\text{Ber}(1/2)$ the channel output of the binary channel $Y_R=X_1\oplus N\oplus Z$ is independent of the channel input $X_1$ \cite{Cover:book}. Then, the capacity region in (\ref{eq:capreg}) is given by
\begin{IEEEeqnarray}{rCl}
\mathcal{C}=&\sup\,&\{I(X_1;\hat{Y}_R|UZ):R_1\geq I(Y_R;\hat{Y}_R|UZ)\},\nonumber
\end{IEEEeqnarray}
where the supremum is taken over the set of pmf's in the form
\begin{IEEEeqnarray}{rCl}
p(u,x_1)p(z)p(y_R|x_1,z)p(\hat{y}_R|y_R,u).\nonumber
\end{IEEEeqnarray}

Let us define $\bar{Y}\triangleq X_1\oplus N$. The capacity is equivalent to
\begin{IEEEeqnarray}{rCl}\label{eq:CapCutconst}
\mathcal{C}=&\sup\,&\{I(X_1;\hat{Y}_R|UZ):H(\bar{Y}|\hat{Y}UZ)\geq H(\bar{Y}|U)- R_1\},\nonumber
\end{IEEEeqnarray}
over the joint pmf's of the form
\begin{IEEEeqnarray}{rCl}
p(u,x_1)p(z)p(\bar{y}|x_1)p(\hat{y}_R|\bar{y},u,z),
\end{IEEEeqnarray}
where we have used the fact that $\bar{Y}$ is independent from $Z$.

For any joint distribution for which $0\leq H(\bar{Y}|U)\leq R_1$, the constraint in (\ref{eq:CapCutconst}) is also satisfied.  It follows from the Markov chain $X_1-\bar{Y}-\hat{Y}_R$ given $U,Z$, and the data processing inequality, that
 \begin{IEEEeqnarray}{rCl}
\mathcal{C}&\leq&\max_{p(u,x_1)}\{I(X_1;\bar{Y}|ZU):H(\bar{Y}|U)\leq R_1\}\\
&=&\max_{p(u,x_1)}\{H(\bar{Y}|U)-h_2(\delta):H(\bar{Y}|U)\leq R_1\}\nonumber\\
&\leq&R_1-h_2(\delta).\nonumber
\end{IEEEeqnarray}

We next consider the joint distributions for which $R_1\leq H(\bar{Y}|U)$. Let $p(u)=\text{Pr}[U=u]$ for $u=1,..,|\mathcal{U}|$, and we can write
\begin{IEEEeqnarray}{rCl}\label{eq:Inf1}
I(&X_1&;\hat{Y}_R|UZ)=H(X_1|U)-\sum_{u}p(u)H(X_1|\hat{Y}_RZu),
\end{IEEEeqnarray}
and
\begin{IEEEeqnarray}{rCl}\label{eq:Inf2}
I(Y_R;\hat{Y}_R|UZ)&\stackrel{(a)}{=}&I(\bar{Y};\hat{Y}_R|UZ)\nonumber\\
&\stackrel{(b)}{=}&H(\bar{Y}|U)-\sum_{u}p(u)H(\bar{Y}|\hat{Y}_RZu),
\end{IEEEeqnarray}
where $(a)$ follows from the definition of $\bar{Y}$, and $(b)$ follows from the independence of $Z$ from $\bar{Y}$ and $U$.

For each $u$, the channel input $X_1$ corresponds to a binary random variable $X_u\sim\text{Ber}(\nu_u)$, where $\nu_u\triangleq\text{Pr}[X_1=1|U=u]=p(1|u)$ for $u=1,...,|\mathcal{U}|$. The channel output for each $X_u$ is given by $\bar{Y}_u=X_u\oplus N$. We denote by $q_u\triangleq\text{Pr}[Y_u=1]=\text{Pr}[Y_R=1|U=u]$. Similarly, we define $\hat{Y}_u$ as $\hat{Y}_R$ for each $u$ value. Note that for each $u$, $X_u-\bar{Y}_u-\hat{Y}_u$ form a Markov chain.

Then, we have $H(X_1|u)=h_2(\nu_u)$ and $H(\bar{Y}|u) = h_2(\delta\star \nu_u)$.
We define $s_u\triangleq H(\bar{Y}|\hat{Y}_RZu)$, such that $0\leq s_u\leq H(\bar{Y}_u)$.
Substituting (\ref{eq:Inf1}) and (\ref{eq:Inf2}) in  (\ref{eq:CapCutconst}) we have
\begin{IEEEeqnarray}{rCl}
\mathcal{C}&=&\max_{p(u,x_1)p(\hat{y}_R|y_R,u)}[H(X_1|U)-\sum_{u}p(u)H(X_1|\hat{Y}_RZu)]\nonumber\\
&& \text{s.t. }R_1\geq H(\bar{Y}|U)-\sum_{u}p(u)H(\bar{Y}|\hat{Y}_RZu)\nonumber\\
&\stackrel{(a)}{=}&\max_{p(u,x_1), }[H(X_1|U)-\sum_{u}p(u)F_{T_{XY}}(q_u,s_u)]\nonumber\\
&& \text{s.t. }R_1\geq H(\bar{Y}|U) -\sum_{u}p(u)s_u,\,0\leq s_u\leq H(\bar{Y}_u)\nonumber\\
&\stackrel{(b)}{=}&\max_{p(u,x_1)}[H(X_1|U)-\sum_{u}p(u)h_2(\delta\star h_2^{-1}(s_u))]\nonumber\\
&&\text{s.t. }R_1\geq H(\bar{Y}|U)-\sum_{u}p(u)s_u\nonumber, \;0\leq s_u\leq H(\bar{Y}_u),\nonumber\\
&\stackrel{(c)}{\leq}&\max_{p(u,x_1)}H(X_1|U)-h_2\left(\delta\star h_2^{-1}\left(\sum_{u}p(u)s_u\right)\right)\nonumber\\
&&\text{s.t. }\sum_{u}p(u)s_u\geq H(\bar{Y}|U)-R_1,\nonumber
\end{IEEEeqnarray}
where $(a)$ follows from the definition of $F_{T_{XY}}(q,s)$ for channel $\bar{Y}_u=X_u\oplus N$, which for each $u$ has a matrix $T_{XY}$ as in (\ref{eq:TxyBinMat}), $(b)$ follows from the expression of $F_{T_{XY}}(q,s)$ for the binary channel $T_{XY}$ in Lemma \ref{lem:binaryF}, $(c)$ follows from noting that $-h_2(\delta\star h_2^{-1}(s_u))$ is concave on $s_u$ from Lemma \ref{lemm:fuConvex} and applying Jensen's inequality.  We also drop the conditions on $s_u$, which can only increase $\mathcal{C}$.

 Then, similarly to the proof of Lemma \ref{lemm:CF}, we have $H(\bar{Y}|U)\geq H(\bar{Y}|UV)=H(X_1|U)$, and we can upper bound the capacity as follows
\begin{IEEEeqnarray}{rCl}
\mathcal{C}&\leq&\max_{p(x_1,u)}\left[H(X_1|U)-h_2\left(\delta\star h_2^{-1}\left(\sum_{u}p(u)s_u\right)\right)\right]\nonumber\\
&&\text{s.t. }\sum_{u}p(u)s_u\geq \max\{H(X_1|U),R_1\}-R_1\nonumber\\
&\leq&\max_{0\leq \alpha\leq 1}\alpha-h_2(\delta\star h_2^{-1}( \max\{\alpha,R_1\}-R_1)),
\end{IEEEeqnarray}
where we have defined $\alpha\triangleq H(X_1|U)$.

The optimization problem can be solved similarly to the proof in Appendix \ref{app:CF} as follows.
If $0\leq \alpha\leq R_1$, we have $\bar{s}\geq0$ and \begin{IEEEeqnarray}{rCl}\label{eq:g_alpha_0}
\mathcal{C}\leq \max_{0\leq \alpha\leq R_1}\alpha-h_2(\delta)=R_1-h_2(\delta).
\end{IEEEeqnarray}
For $R_1\leq \alpha\leq 1$, we have
\begin{IEEEeqnarray}{rCl}\label{eq:g_alpha_R0}
\mathcal{C}\leq\max_{R_1\leq\alpha\leq 1}\alpha-h_2(\delta\star h_2^{-1}( \alpha-R_1)).
\end{IEEEeqnarray}

Then, it follows from a scaled version of Proposition \ref{prop:galphaIncreasing} that the upper bound is maximized for $\alpha=1$. Then, by noticing that (\ref{eq:g_alpha_0}) corresponds to the value of the bound in (\ref{eq:g_alpha_R0}) for $\alpha=R_1$, it follows that
 \begin{IEEEeqnarray}{rCl}
\mathcal{C}&\leq&1-h_2(\delta\star h_2^{-1}(1-R_1)).
\end{IEEEeqnarray}
This bound is achievable by CF. This completes the proof.

\section{Proof of the Cut-Set Bound Optimality Conditions}\label{app:CutSetProof}
Cases $1$ and $2$ are straightforward since under these
assumptions, the ORC-D studied here becomes a
particular case of the channel models in
\cite{ElGamal:abbas2005} and \cite{Kim2007CapacityClassRelay},
respectively.

To prove Case $3$ we use the following arguments. For any channel input distribution to the
ORC-D, we have
\begin{IEEEeqnarray}{rCl}
I(X_1;Y_R|Z)&=&H(X_1|Z)-H(X_1|Y_R,Z)\nonumber\\
&\geq&H(X_1)-H(X_1|Y_R)\\
&=&I(X_1;Y_R),\nonumber
\end{IEEEeqnarray}
where we have used the independence of $X_1$ and $Z$, and the fact that
conditioning reduces entropy. Then, the condition
$\max_{p(x_1)}I(X_1;Y_R)\geq R_1$, implies
$\max_{p(x_1)}I(X_1;Y_R|Z)\geq R_1$; and hence, the cut-set bound is
given by $R_{\mathrm{CS}}=R_2+R_1$, which is achievable by DF scheme.

In Case $4$, the cut-set bound is given by
$R_2+\min\{R_1,I(\bar{X}_1;\bar{Y}_R|Z)\}=R_2+I(\bar{X}_1;\bar{Y}_R|Z)$
since $R_1\geq H(\bar{Y}_R|Z)$. CF achieves the capacity by letting
 $X_1$ be distributed with $\bar{p}(x_1)$, and choosing
$\hat{Y}_R=\bar{Y}_R$. This choice is always possible as the CF
constraint
\begin{IEEEeqnarray}{rCl}
R_1\geq
I(\hat{Y}_R;\bar{Y}_R|Z)=H(\bar{Y}_R|Z)-H(\bar{Y}_R|Z,\hat{Y}_R)=H(\bar{Y}_R|Z),\nonumber
\end{IEEEeqnarray}
always holds. Then, the
achievable rate for CF is
$R_{\mathrm{CF}}=R_2+I(\bar{X}_1;\hat{Y}_R|Z)=R_2+I(\bar{X}_1;\bar{Y}_R|Z)$, which is the capacity.

\end{appendices}

\bibliographystyle{ieeetran}
\balance
\bibliography{ref}
\vspace{-1cm}
\begin{IEEEbiographynophoto}
{I\~naki Estella Aguerri}
I\~naki Estella Aguerri received his B.Sc. and M.Sc. degrees in Telecommunication engineering from Universitat Polit\`ecnica de Catalunya (UPC), Barcelona, Spain, in 2008 and 2011, respectively. He received his Ph.D. degree from Imperial College London, London, UK, in January 2015. In November 2014, he joined the Mathematical and Algorithmic Sciences Lab, France Research Center, Huawei Technologies Co. Ltd., Boulogne-Billancourt, France. From 2008 to 2013, he was research assistant at Centre Tecnològic de Telecomunicacions de Catalunya (CTTC) in Barcelona. His primary research interests are in the areas of information theory and communication theory, with special emphasis on joint source-channel coding and cooperative communications for wireless networks.
\end{IEEEbiographynophoto}
\vspace{-1cm}

\begin{IEEEbiographynophoto}
{Deniz Gunduz}
Deniz Gunduz [S'03-M'08-SM'13] received the B.S. degree in electrical and electronics engineering from METU, Turkey in 2002, and the M.S. and Ph.D. degrees in electrical engineering from NYU Polytechnic School of Engineering in 2004 and 2007, respectively. After his PhD he served as a postdoctoral research associate at the Department of Electrical Engineering, Princeton University, and as a consulting assistant professor at the Department of Electrical Engineering, Stanford University. He also served as a research associate at CTTC in Barcelona, Spain. Since September 2012 he is a Lecturer in the Electrical and Electronic Engineering Department of Imperial College London, UK. He also held a visiting researcher position at Princeton University from November 2009 until November 2011. 
 
Dr. Gunduz is an Associate Editor of the IEEE Transactions on Communications, and an Editor of the IEEE Journal on Selected Areas in Communications (JSAC) Series on Green Communications and Networking. He is the recipient of a Starting Grant of the European Research Council, the 2014 IEEE Communications Society Best Young Researcher Award for the Europe, Middle East, and Africa Region, and the Best Student Paper Award at the 2007 IEEE International Symposium on Information Theory (ISIT). He is serving as a co-chair of the IEEE Information Theory Society Student Committee, and the co-director of the Imperial College Probability Center. He is the General Co-chair of the 2016 IEEE Information Theory Workshop. He served as a Technical Program Chair of the Network Theory Symposium at the 2013 and 2014 IEEE Global Conference on Signal and Information Processing (GlobalSIP), and General Co-chair of the 2012 IEEE European School of Information Theory (ESIT). His research interests lie in the areas of communication theory and information theory with special emphasis on joint source-channel coding, multi-user networks, energy efficient communications and privacy in cyber-physical systems.

\end{IEEEbiographynophoto}

\newpage

\end{document}